%% file: paper.tex
\documentclass{llncs}
\usepackage{graphicx}
\usepackage{amsmath}
\usepackage{amssymb}
\usepackage{algorithm}
\usepackage[noend]{algorithmic}
\usepackage[tight,TABTOPCAP]{subfigure}
\usepackage{times}
\usepackage{picinpar}

\input{macros}
\begin{document}

\title{Classifying and Propagating Parity Constraints (extended version)\thanks{The original version of the paper~\cite{LJN:CP2012} has been presented in the 18th International Conference on Principle and Practice of Constraint Programming, CP 2012. An earlier version of the extended version has been presented for CP 2012 reviewers. This revised version uses proof techniques from \cite{LJN:ICTAI2012full}.}}

\author{Tero Laitinen, Tommi Junttila, and Ilkka Niemel\"a}
\institute{Aalto University\\
  Department of Information and Computer Science\\
  PO Box 15400, FI-00076 Aalto, Finland\\
  \email{\{Tero.Laitinen,Tommi.Junttila,Ilkka.Niemela\}@aalto.fi}
}

\maketitle

\begin{abstract}
Parity constraints, common in application domains such
as circuit verification, bounded model checking, and logical
cryptanalysis, are not necessarily most efficiently solved if translated
into conjunctive normal form. Thus, specialized parity reasoning
techniques have been developed in the past for propagating parity
constraints. This paper studies the questions of deciding whether
unit propagation or equivalence reasoning is enough to achieve full
propagation in a given parity constraint set. Efficient approximating
tests for answering these questions are developed. It is also shown
that equivalence reasoning can be simulated by unit propagation by
adding a polynomial amount of redundant parity constraints to the problem.
It is proven that without using additional variables, an exponential
number of new parity constraints would be needed in the worst case.
The presented classification and propagation methods are evaluated
experimentally.

\end{abstract}


%
%
\input{introduction}

\input{preliminaries}
\input{dpllxor}

\input{treelike}

\input{subst}

\input{conclusions}

\bibliographystyle{splncs}
\bibliography{paper}

\appendix
\input{proofs}

\end{document}

%% file: macros.tex
\newcommand{\ToolName}[1]{\textsf{\small #1}}

\newcommand{\CGraph}{G}
\newcommand{\CGNodes}{V}
\newcommand{\CGVNodes}{V_\textup{vars}}
\newcommand{\CGCNodes}{V_\textup{clauses}}
\newcommand{\CGEdges}{E}
\newcommand{\CGLab}{L}
\newcommand{\Trivium}{{\sc Trivium}}

\newcommand{\UnConflictRule}{\ensuremath{\textsf{\small Conflict}}}
\newcommand{\UnConflictRuleP}{\ensuremath{\textsf{Conflict}}}
\newcommand{\TerBinRuleP}{\ensuremath{\textsf{$\oplus$-Unit$^3$}}}
\newcommand{\TerBinRule}{\ensuremath{\textsf{$\oplus${\small-Unit}$^3$}}}
\newcommand{\BinUnRuleP}{\ensuremath{\textsf{$\oplus$-Unit$^2$}}}
\newcommand{\BinUnRule}{\ensuremath{\textsf{$\oplus${\small-Unit}$^2$}}}
\newcommand{\TerUnRuleP}{\ensuremath{\textsf{$\oplus$-Imply}}}
\newcommand{\TerUnRule}{\ensuremath{\textsf{$\oplus$\small-Imply}}}
\newcommand{\ConflictRuleP}{\ensuremath{\textsf{$\oplus$-Conflict}}}
\newcommand{\ConflictRule}{\ensuremath{\textsf{$\oplus$\small-Conflict}}}

\newcommand{\EC}{{\ensuremath{\textup{\small\textsf{EC}}}}}
\newcommand{\XX}{{\oplus}}
\newcommand{\xor}{\oplus}

\newcommand{\parity}[1]{\ensuremath{p_{#1}}}

\newcommand{\ECsys}{\textsc{\textrm EC}}
\newcommand{\ECDeriv}{\mathrel{\vdash_{\textup{ec}}}}

\newcommand{\SUBST}{{\ensuremath{\textup{\small\textsf{Subst}}}}}
\newcommand{\SUBSTDeriv}{\mathrel{\vdash_{\textup{\textsf{Subst}}}}}
\newcommand{\SUBSTDerivNot}{\mathrel{{\not\,\vdash}_{\textup{\textsf{Subst}}}}}

\newcommand{\UP}{{\ensuremath{\textup{\small\textsf{UP}}}}}
\newcommand{\UPDeriv}{\mathrel{\vdash_{\textup{up}}}}
\newcommand{\UPDerivNot}{\mathrel{{\not\,\vdash}_{\textup{up}}}}

\newcommand{\set}[1]{\left\{ #1 \right\}}
\newcommand{\Set}[1]{\set{#1}}

\newcommand{\Setdef}[2]{\left\{{#1}\mid{#2}\right\}}
\newcommand{\Tuple}[1]{\ensuremath{{\left\langle{#1}\right\rangle}}}

\newcommand{\True}{\top}
\newcommand{\False}{\bot}
\newcommand{\T}{\top}
\newcommand{\F}{\bot}

\newcommand{\X}{\oplus}

\newcommand{\TA}{\tau}
\newcommand{\Equiv}{\equiv}
\newcommand{\Implies}{\Rightarrow}

\newcommand{\CNF}[1]{\operatorname{cnf}(#1)}

\newcommand{\XC}{D}

\newcommand{\substitution}[2]{#1/#2}
\newcommand{\simplification}[3]{\ensuremath{#1 \left[ \substitution{#2}{#3}\right] } }

\newcommand{\xorderivation}[1]{\ensuremath{\pi_{\mathit{#1}}}}
\newcommand{\ECderivation}{\ensuremath{\pi}}
\newcommand{\UPderivation}{\xorderivation{up}}
\newcommand{\SUBSTderivation}{\ensuremath{\pi}}

\newcommand{\unitruleP}{\ensuremath{\oplus\mbox{-Unit}^+}}
\newcommand{\unitruleN}{\ensuremath{\oplus\mbox{-Unit}^-}}
\newcommand{\eqvruleP}{\ensuremath{\oplus\mbox{-Eqv}^+}}
\newcommand{\eqvruleN}{\ensuremath{\oplus\mbox{-Eqv}^-}}

\newcommand{\clauses}{\ensuremath{\phi}}
\newcommand{\orclauses}{\clauses_{\textup{or}}}
\newcommand{\xorclauses}{\clauses_{\textup{xor}}}
\newcommand{\xorclausesI}[1]{\clauses_{\textup{xor},#1}}

\newcommand{\Var}{x}
\newcommand{\VarsOf}[1]{\operatorname{vars}(#1)}
\newcommand{\LitsOf}[1]{\operatorname{lits}(#1)}

\newcommand{\Models}{\ensuremath{\models}}
\newcommand{\ModelsNot}{\ensuremath{\mathrel{\not\models}}}
\newcommand{\Conseq}{\ensuremath{\models}}

\newcommand{\defrule}[2]{
\begin{tabular}{@{}c@{}}
        \ensuremath{#1} \\
    \hline
        \ensuremath{#2} 
    \end{tabular}}

\newcommand{\defrulecond}[3]{
\makebox{%
\begin{tabular}{@{}c@{}}
  \ensuremath{#1} \\
  provided that \ensuremath{#3}\\
  \hline
  \ensuremath{#2} 
\end{tabular}%
}}

\newcommand{\inferencerule}[3]{\defrule{{#1}\hspace{5mm}{#2}}{#3}}
\newcommand{\inferencerulen}[2]{\defrule{#1}{#2}}
\newcommand{\inferencerulenc}[3]{\defrulecond{#1}{#2}{#3}}

\newcommand{\Comment}[1]{\textsf{\footnotesize/*{#1}*/}}
\newcommand{\While}{\KW{while}}
\newcommand{\For}{\KW{for}}
\newcommand{\Each}{\KW{each}}
\newcommand{\If}{\KW{if}}
\newcommand{\Elif}{\KW{else if}}
\newcommand{\Else}{\KW{else}}
\newcommand{\Break}{\KW{break}}
\newcommand{\Return}{\KW{return}}
\newcommand{\KW}[1]{\textsf{\footnotesize{}{#1}}}
\newcommand{\Confl}{\mathit{confl}}
\newcommand{\IL}{\hat{l}}
\newcommand{\AL}{\tilde{l}}
\newcommand{\MyForm}{\ensuremath{\phi}}
\newcommand{\vars}{\operatorname{vars}}
\newcommand{\lits}{\operatorname{lits}}

\newcommand{\diamondcycles}{\ensuremath{D(n)}}

%% file: introduction.tex
\section{Introduction}

Encoding a problem instance in conjunctive normal form (CNF) allows very
efficient Boolean constraint propagation and conflict-driven clause learning
(CDCL) techniques.  
This has contributed to the success of propositional satisfiability (SAT)
solvers (see e.g.~\cite{Handbook:CDCL}) in a number of industrial
application domains. 
On the other hand,
an instance consisting only of parity (xor) constraints can be solved in
polynomial time using Gaussian elimination
but
CNF-based solvers relying only on basic Boolean constraint propagation
tend to scale poorly on the straightforward CNF-encoding of the instance.
To handle CNF instances including parity constraints,
common in application domains such as circuit verification, bounded model checking, and logical cryptanalysis, 
several approaches have been developed~\cite{Li:AAAI2000,Li:IPL2000,BaumgartnerMassacci:CL2000,Li:DAM2003,HeuleMaaren:SAT2004,HeuleEtAl:SAT2004,Chen:SAT2009,SoosEtAl:SAT2009,LJN:ECAI2010,Soos,LJN:ICTAI2011,LJN:SAT2012}.
These approaches extend CNF-level SAT solvers by
implementing different forms of constraint propagation for parity constraints,
ranging from plain unit propagation
via equivalence reasoning
to Gaussian elimination.
Compared to unit propagation,
which has efficient implementation techniques,
equivalence reasoning and Gaussian elimination allow stronger propagation but
are computationally much more costly.

In this paper our main goal is not to design new inference rules and data structures for propagation engines,
but to develop
(i) methods for analyzing the structure of parity constraints in order to detect how powerful a parity reasoning engine is needed to achieve full forward propagation,
and
(ii) translations that allow unit propagation to simulate equivalence reasoning.
We first present a method for detecting parity constraint sets for which unit
propagation achieves full forward propagation.
For instances that do not fall into this category,
we show how to extract easy-to-propagate parity constraint parts 
so that they can be handled by unit propagation and
the more powerful reasoning engines can take care of the rest.
We then describe a method for detecting parity constraint sets
for which equivalence reasoning achieves full forward propagation.
By analyzing the set of parity constraints as a constraint graph, we can
characterize equivalence reasoning using the cycles in the graph.
By enumerating these cycles and adding a new linear combination of the original
constraints for each such cycle to the instance,
we can achieve an instance in which unit propagation simulates equivalence reasoning.
As there may be an exponential number of such cycles,
we develop another translation to simulate equivalence reasoning with
unit propagation.
The translation is polynomial as new variables are introduced;
we prove that if introduction of new variables is not allowed,
then there are instance families for which polynomially sized simulation
translations do not exist.
This translation can be optimized significantly by adding only a selected
subset of the new parity constraints.
Even though the translation is meant to simulate equivalence reasoning with
unit propagation, it can augment the strength of equivalence reasoning if
equivalence reasoning does not achieve full forward propagation on the original
instance. 
The presented detection and translation methods are evaluated experimentally
on large sets of benchmark instances.
The proofs of lemmas and theorems can be found in the appendix.
%

%% file: preliminaries.tex
\section{Preliminaries}

An \emph{atom} is either a propositional variable or
the special symbol $\top$ which denotes the constant ``true''.
A \emph{literal} is an atom $A$ or its negation $\neg A$;
we identify $\neg\top$ with $\bot$ and $\neg \neg A$ with $A$.
A traditional, non-exclusive \emph{or-clause} is a disjunction
$l_1 \vee \dots \vee l_n$ of literals.
Parity constraints are formally presented with xor-clauses:
an {\it xor-clause} is an expression of form $ l_1 \oplus \dots \oplus l_n$,
where $l_1,\dots,l_n$ are literals and the symbol $\oplus$ stands for the
exclusive logical or.
In the rest of the paper,
we implicitly assume that each xor-clause is in a \emph{normal form}
such that (i) each atom occurs at most once in it, and
(ii) all the literals in it are positive.
The unique (up to reordering of the atoms) normal form
for an xor-clause can be obtained by applying the following rewrite rules
in any order until saturation:
(i) ${{\neg A} \oplus C} \leadsto {A \oplus \top \oplus C}$, and
(ii) ${A \oplus A \oplus C} \leadsto C$,
where $C$ is a possibly empty xor-clause and $A$ is an atom.
For instance,
the normal form of ${\neg x_1} \oplus x_2 \oplus x_3 \oplus x_3$ is 
$x_1 \oplus x_2 \oplus \top$,
while the normal form of $x_1 \oplus x_1$ is the empty xor-clause $()$.
We say that an xor-clause is unary/binary/ternary if
its normal form has one/two/three variables, respectively.
We will identify $x \oplus \top$ with the literal $\neg x$.
For convenience, we can represent xor-clauses in equation form
$x_1 \X ... \X x_k \Equiv \parity{}$ with $\parity{} \in \Set{\F,\T}$;
e.g.,
$x_1 \oplus x_2$ is represented with $x_1 \oplus x_2 \Equiv \T$
and
$x_1 \oplus x_2 \oplus \top$ with $x_1 \oplus x_2 \Equiv \F$.
The straightforward CNF translation of an xor-clause $\XC$
is denoted by $\CNF{\XC}$;
for instance,
$\CNF{x_1 \X x_2 \X x_3 \X \T} =
 (\neg x_1 \lor \neg x_2 \lor \neg x_3) \land 
 (\neg x_1 \lor x_2 \lor x_3) \land 
 (x_1 \lor \neg x_2 \lor x_3) \land 
 (x_1 \lor x_2 \lor \neg x_3)$.
A {\it clause} is either an or-clause or an xor-clause.

A {\it truth assignment} $\TA$ is a set of literals such that
$\top \in \TA$ and $\forall l \in \TA : {\neg l} \notin \TA$.
We define the ``satisfies'' relation $\models$ between a truth assignment
$\TA$ and logical constructs as follows:
(i) if $l$ is a literal,
then $\TA \models l$ iff $l \in \TA$,
(ii) if $C = (l_1 \lor \dots \lor l_n)$ is an or-clause,
then $\TA \models C$ iff $\TA \models l_i$ for some $l_i \in \set{l_1,\ldots,l_n}$,
and
(iii)
if $C = (l_1 \oplus \dots \oplus l_n)$ is an xor-clause,
then $\TA \models C$ iff
$\TA$ is total for $C$ (i.e.~$\forall 1 \le i \le n : {l_i \in \TA} \lor {\neg l_i \in \TA}$)
and
$\TA \models l_i$ for an odd number of literals of $C$.
Observe that no truth assignment satisfies
the empty or-clause $()$ or the empty xor-clause $()$,
i.e.~these clauses are synonyms for $\bot$.

A \emph{cnf-xor formula} $\MyForm$ is a conjunction of clauses,
expressible as a conjunction
\begin{equation}
\MyForm = \orclauses \land \xorclauses,
\end{equation}
where $\orclauses$ is a conjunction of or-clauses and $\xorclauses$ is a conjunction of xor-clauses.
A truth assignment $\TA$ \emph{satisfies} $\MyForm$,
denoted by $\TA \models \MyForm$,
if it satisfies each clause in it;
$\MyForm$ is called \emph{satisfiable} if there exists such a truth assignment
satisfying it, and \emph{unsatisfiable} otherwise.
The \emph{cnf-xor satisfiability} problem studied in this paper is to decide
whether a given cnf-xor formula has a satisfying truth assignment.
A formula $\MyForm'$ is a \emph{logical consequence} of a formula $\MyForm$,
denoted by $\MyForm \Conseq \MyForm'$,
if $\TA \models \MyForm$ implies $\TA \models \MyForm'$
for all truth assignments $\TA$ that are total for $\MyForm$ and $\MyForm'$.
The set of variables occurring in a formula $\MyForm$ is denoted by
$\vars(\MyForm)$,
and 
$\lits(\MyForm) = \Setdef{x, \neg x}{x \in \vars(\MyForm)}$ is
the set of literals over $\vars(\MyForm)$.
We use $\simplification{C}{A}{D}$ to denote the (normal form) xor-clause that is
identical to $C$ except that all occurrences of the atom $A$ in $C$
are substituted with $D$ once.
For instance,
$\simplification{(x_1 \oplus x_2 \oplus x_3)}{x_1}{(x_1 \oplus x_3)} =
{x_1 \oplus x_3 \oplus x_2 \oplus x_3} =
{x_1 \oplus x_2}$.

%% file: dpllxor.tex
\subsection{The DPLL(XOR) framework}

To separate parity constraint reasoning from the CNF-level reasoning,
we apply the recently introduced
DPLL(XOR) framework \cite{LJN:ECAI2010,LJN:ICTAI2011}.
The idea in the DPLL(XOR) framework
for satisfiability solving of cnf-xor formulas
$\MyForm = \orclauses \land \xorclauses$
is similar to that in the DPLL($T$) framework for solving
satisfiability of quantifier-free first-order formulas modulo
a background theory $T$
(SMT, see e.g.\ \cite{NieuwenhuisEtAl:JACM06,Handbook:SMT}).
In DPLL(XOR),
see Fig.~\ref{fig:DPLLXOR} for a high-level pseudo-code,
one employs a conflict-driven clause learning (CDCL) SAT solver
(see e.g.~\cite{Handbook:CDCL})
to search for a satisfying truth assignment $\TA$
over all the variables in $\MyForm = \orclauses \land \xorclauses$.
%
The CDCL-part takes care of the usual unit clause propagation on the cnf-part
$\orclauses$ of the formula (line 4 in Fig.~\ref{fig:DPLLXOR}),
conflict analysis and non-chronological backtracking (line 15--17),
and
heuristic selection of decision literals (lines 19--20) which
extend the current partial truth assignment $\TA$ towards a total one.

\begin{figure}[bt]
{\small
\begin{tabbing}
99.\={mm}\={mm}\={mm}\=\kill
solve($\MyForm = {\orclauses \land \xorclauses}$):\\
1.\>initialize xor-reasoning module $M$ with $\xorclauses$\\
2.\>$\TA = \Tuple{}$\qquad\Comment{the truth assignment}\\
3.\>\While{} true:\\
4.\>\>$(\TA',\Confl) = \textsc{unitprop}(\orclauses,\TA)$\quad\Comment{unit propagation}\\
5.\>\>\If{} not $\Confl$:\qquad\Comment{apply xor-reasoning}\\
6.\>\>\>\For{} \Each{} literal $l$ in $\TA'$ but not in $\TA$: $M$.\textsc{assign}($l$)\\
7.\>\>\>$(\IL_1,...,\IL_k) = M.\textsc{deduce}()$\\
8.\>\>\>\For{} $i=1$ to $k$:\\
9.\>\>\>\>$C = M.\textsc{explain}(\IL_i)$\\
10.\>\>\>\>\If{} $\IL_i = \False$ or ${\neg \IL_i} \in \TA'$: $\Confl = C$, \Break{}\\
11.\>\>\>\>\Elif{} $\IL_i \notin \TA'$: add $\IL_i^{C}$ to $\TA'$\\
12.\>\>\>\If{} $k > 0$ and not $\Confl$:\\
13.\>\>\>\>$\TA = \TA'$; continue\quad\Comment{unit propagate further}\\
14.\>\>let $\TA = \TA'$\\
15.\>\>\If{} $\Confl$:\qquad\Comment{standard Boolean conflict analysis}\\
16.\>\>\>analyze conflict, learn a conflict clause\\
17.\>\>\>backjump or return ``unsatisfiable'' if not possible\\
18.\>\>\Else{}: \\
19.\>\>\>add a heuristically selected unassigned literal in $\MyForm$ to $\TA$\\
20.\>\>\>or return ``satisfiable'' if no such variable exists
\end{tabbing}%
}%
\vspace{-2mm}
\caption{The essential skeleton of the DPLL(XOR) framework}%
\label{fig:DPLLXOR}%
\end{figure}%

To handle the parity constraints in the xor-part $\xorclauses$,
an \emph{xor-reasoning module} $M$ is coupled with the CDCL solver.
The values assigned in $\TA$ to the variables in $\VarsOf{\xorclauses}$
by the CDCL solver are communicated as \emph{xor-assumption literals}
to the module (with the \textsc{assign} method on line 6 of the pseudo-code).
If $\AL_1,...,\AL_m$ are the xor-assumptions communicated to the module so far,
then the \textsc{deduce} method (invoked on line 7) of the module
is used to deduce a (possibly empty) list of \emph{xor-implied literals}
$\IL$ that are logical consequences of the xor-part $\xorclauses$ and
xor-assumptions,
i.e.~literals for which ${\xorclauses \land \AL_1 \land ... \land \AL_m} \Models \IL$ holds.
These xor-implied literals can then be added to the current truth assignment
$\TA$ (line 11) and the CDCL part invoked again to perform unit clause
propagation on these.
The conflict analysis engine of CDCL solvers requires that
each implied (i.e.~non-decision) literal has an \emph{implying clause},
i.e.~an or-clause that forces the value of the literal by unit propagation
on the values of literals appearing earlier in
the truth assignment (which at the implementation level is a sequence of
literals instead of a set).
For this purpose
the xor-reasoning module has a method \textsc{explain} that,
for each xor-implied literal $\IL$,
gives an or-clause $C$ of form ${l'_1 \land ... \land l'_k} \Implies \IL$,
i.e.~${\neg l'_1} \lor ... \lor {\neg l'_k} \lor \IL$,
such that
(i) $C$ is a logical consequence of $\xorclauses$, 
and
(ii) $l'_1,...,l'_k$ are xor-assumptions made or xor-implied literals returned
before $\IL$.
An important special case occurs when the ``false'' literal $\False$
is returned as an xor-implied literal (line 10),
i.e.~when an \emph{xor-conflict} occurs;
this implies that ${\xorclauses \land \AL_1 \land ... \land \AL_m}$
is unsatisfiable.
In such a case, the clause returned by the \textsc{explain} method
is used as the unsatisfied clause $\Confl$ initiating the conflict analysis
engine of the CDCL part (lines 10 and 15--17).
\emph{In this paper,
we study the process of deriving xor-implied literals} and
will not describe in detail how implying or-clauses are computed;
the reader is referred to~\cite{LJN:ECAI2010,LJN:ICTAI2011,LJN:SAT2012}.

Naturally,
there are many \emph{xor-module integration strategies}
that can be considered in addition to the one described in
the above pseudo-code.
For instance,
if one wants to prioritize xor-reasoning,
the xor-assumptions can be given one-by-one instead.
Similarly, if CNF reasoning is to be prioritized, the xor-reasoning module
can lazily compute and return the xor-implied literals one-by-one only
when the next one is requested.

In addition to our previous work~\cite{LJN:ECAI2010,LJN:ICTAI2011,LJN:SAT2012},
also cryptominisat~\cite{SoosEtAl:SAT2009,Soos} can be seen to follow
this framework.

%% file: treelike.tex
\section{Unit Propagation}

We first consider the problem of deciding, given an xor-clause conjunction,
whether the elementary unit propagation technique is enough
for always deducing all xor-implied literals.
As we will see, this is actually the case for many ``real-world'' instances.
The cnf-xor instances having such xor-clause conjunctions are probably best handled either by translating the xor-part into CNF or with unit propagation algorithms on parity constraints~\cite{Chen:SAT2009,SoosEtAl:SAT2009,LJN:SAT2012}
instead of more complex xor-reasoning techniques.

To study unit propagation on xor-clauses,
we introduce a very simple xor-reasoning system ``\UP{}''
that can only deduce the same xor-implied literals as CNF-level unit
propagation would on the straightforward CNF translation of the xor-clauses. 
To do this, \UP{} implements the deduction system with 
the inference rules shown in Fig.~\ref{Fig:UPRules}.
A {\it UP-derivation} from a conjunction of xor-clauses $\psi$ is a sequence of
xor-clauses $\XC_1, \dots, \XC_n$ where each $\XC_i$ is either (i) in $\psi$, or
(ii) derived from two xor-clauses $\XC_j$, $\XC_k$ with $1 \le j < k < i$
using the inference rule $\unitruleP$ or $\unitruleN$.
An xor-clause $\XC$ is {\it UP-derivable} from $\psi$,
denoted $\psi \UPDeriv \XC$,
if there exists a UP-derivation from $\psi$ where $\XC$ occurs.
As an example,
let $\xorclauses = (a \X d \X e) \land (d \X c \X f) \land (a \X b \X c)$.
Fig.~\ref{Fig:UPCG}(a)
illustrates a \UP-derivation from
$\xorclauses \land (a) \land (\neg d)$;
as $\neg e$ occurs in it,
$\xorclauses \land (a) \land (\neg d) \UPDeriv \neg e$
and
thus unit propagation can deduce the xor-implied literal $\neg e$
under the xor-assumptions $(a)$ and $(\neg d)$.

\begin{figure}[t]
  \centering
  \small
  \begin{tabular}{l@{\qquad}l}
    \unitruleP:\ 
    $\inferencerule{x}{C}{\simplification{C}{x}{\top}}$
    &
    \unitruleN:\ 
    $\inferencerule{x \oplus \True}{C}{\simplification{C}{x}{\False}}$
    \\
  \end{tabular}%
  \caption{Inference rules of \UP{}; 
          The symbol $x$ is variable and $C$ is an xor-clause.}
  \label{Fig:UPRules}%
\end{figure}

\begin{definition}
  A conjunction $\xorclauses$ of xor-clauses is \emph{\UP-deducible}
  if for all $\AL_1,...,\AL_k,\IL \in \LitsOf{\xorclauses}$
  it holds that
  (i) if
  $\xorclauses \land \AL_1 \land ... \land \AL_k$ is unsatisfiable,
  then
  $\xorclauses \land \AL_1 \land ... \land \AL_k \UPDeriv \F$,
  and
  (ii)
  ${\xorclauses \land \AL_1 \land ... \land \AL_k} \Models \IL$
  implies
  ${\xorclauses \land \AL_1 \land ... \land \AL_k} \UPDeriv \IL$
  otherwise.
\end{definition}

Unfortunately we do not know any easy way of detecting whether
a given xor-clause conjunction is \UP-deducible.
However, as proven next, xor-clause conjunctions that are ``tree-like'',
an easy to test structural property,
are \UP-deducible.
For this, and also later,
we use the quite standard concept of constraint graphs:
the \emph{constraint graph} of an xor-clause conjunction $\xorclauses$
is a labeled bipartite graph $\CGraph = \Tuple{\CGNodes, \CGEdges, \CGLab}$,
where
\begin{itemize}
\item
  the set of vertices $\CGNodes$
  is the disjoint union of
  (i)
    \emph{variable vertices} $\CGVNodes = \VarsOf{\xorclauses}$
    which are graphically represented with circles,
    and
  (ii)
    \emph{xor-clause vertices} $\CGCNodes = \Setdef{\XC}{\text{$\XC$ is an xor-clause in $\xorclauses$}}$ drawn as rectangles,
\item
  $\CGEdges = {\Setdef{\Set{\Var,\XC}}{{\Var \in \CGVNodes} \land {\XC \in \CGCNodes} \land {\Var \in \VarsOf{\XC}}}}$
  are
  the edges connecting the variables
  and
  the xor-clauses in which they occur,
  and
\item
  $\CGLab$ labels each xor-clause vertex $x_1 \X ... \X x_k \Equiv \parity{}$
  with the parity $\parity{}$.
\end{itemize}
A conjunction $\xorclauses$ is \emph{tree-like}
if its constraint graph is a tree
or
a union of disjoint trees.

\begin{example}\label{Ex:Tree1}
  The conjunction
  $(a \X b \X c) \land (b \X d \X e) \land (c \X f \X g \X \T)$
  is tree-like;
  its constraint graph is given in Fig.~\ref{Fig:UPCG}(b).
  On the other hand,
  the conjunction
  $(a \X b \X c) \land (a \X d \X e) \land (c \X d \X f) \land (b \X e \X f)$,
  illustrated in Fig.~\ref{Fig:UPCG}(c), is not tree-like.
\end{example}

\begin{figure}[t]
  \centering
  \begin{tabular}{@{}c@{\quad}c@{\quad}c@{}}
    \includegraphics[scale=0.70]{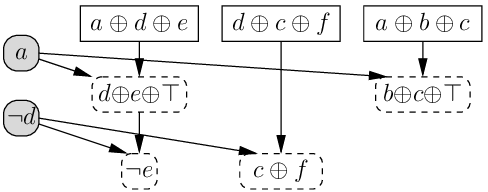} &
    \includegraphics[scale=0.75]{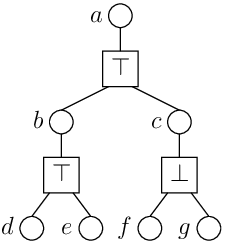} &
    \includegraphics[scale=0.75]{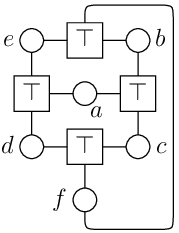}
    \\
    (a) & (b) & (c)%
  \end{tabular}%
  \caption{A \UP-derivation and two constraint graphs}
  \label{Fig:UPCG}
\end{figure}

\begin{theorem}
  If a conjunction of xor-clauses $\xorclauses$ is tree-like,
  then it is \UP-deducible.
  \label{Theorem:Trees}
\end{theorem}
Note that not all \UP-deducible xor-clause constraints are tree-like.
For instance,
$(a \X b) \land (b \X c) \land (c \X a \X \top)$
is satisfiable and \UP-deducible but
not tree-like.
No binary xor-clauses are needed to establish the same, e.g.,
$(a \X b \X c) \land (a \X d \X e) \land (c \X d \X f) \land (b \X e \X f)$
considered in Ex.~\ref{Ex:Tree1}
is satisfiable and \UP-deducible but not tree-like.

%
%
\subsection{Experimental Evaluation}
\label{Sect:TreeClassification}

To evaluate the relevance of this tree-like classification,
we studied the benchmark instances in ``crafted'' and ``industrial/application''
categories of the SAT Competitions 2005, 2007, and 2009 as well as
all the instances in the SAT Competition 2011
(available at \url{http://www.satcompetition.org/}).
We applied the xor-clause extraction algorithm described in~\cite{Soos} to
these CNF instances and found a large number of instances with xor-clauses.
To get rid of some ``trivial'' xor-clauses,
we eliminated unary clauses and binary xor-clauses from each instance
by unit propagation and substitution, respectively.
After this easy preprocessing,
474 instances (with some duplicates due to overlap in the competitions)
having xor-clauses remained.
Of these instances, 61 are tree-like.

As shown earlier,
there are \UP{}-deducible cnf-xor instances that are not tree-like.
To find out whether any of the 413 non-tree-like cnf-xor instances we found
falls into this category,
we applied the following testing procedure to each instance:
randomly generate xor-assumption sets and for each check,
with Gaussian elimination,
whether all xor-implied literals were propagated by unit propagation.
For only one of the 413 non-tree-like cnf-xor instances
the random testing could not prove that it is not \UP-deducible;
thus the tree-like classification seems to work quite well in practice as an
approximation of detecting \UP{}-deducibility.
More detailed results are shown in Fig.~\ref{Fig:Classification}(a).
The columns ``probably \SUBST'' and ``cycle-partitionable'' are explained later.

\begin{figure}[tb]
  \centering
  \begin{tabular}{@{}c@{\quad}c@{}}
    \begin{tabular}{|r@{\ }|@{\ }c@{\ }|@{\ }c@{\ }|@{\ }c@{\ }|@{\ }c@{\ }|}
      \hline
      & \multicolumn{4}{c|}{SAT Competition}\\
      & 2005 & 2007 & 2009 & 2011 \\
      \hline\hline
      instances            & 857 & 376 & 573 & 1200 \\
      with xors            & 123 & 100  & 140 & 111 \\
      \hline
      probably \UP{}       & 19  & 10 & 18 & 15\\
      tree-like            & 19  & 9  & 18 & 15 \\
      \hline
      probably \SUBST{} & 20  & 21 & 52 & 40 \\
      cycle-partitionable  & 20  & 13 & 24 & 40 \\
      \hline
    \end{tabular}
    &
    \begin{minipage}{.50\textwidth}
      \includegraphics[width=.99\textwidth]{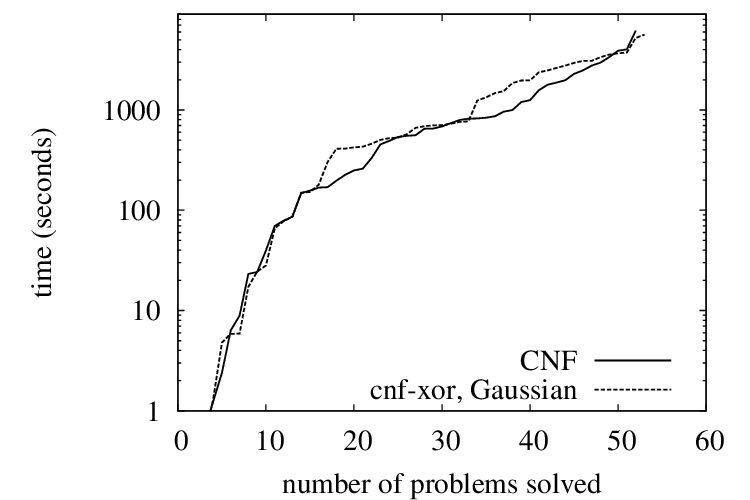}
    \end{minipage}
    \\
    (a) & (b)
  \end{tabular}
  \caption{Instance classification (a), and cryptominisat run-times on tree-like instances (b)}
  \label{Fig:Classification}
\end{figure}

As unit propagation is already complete for tree-like cnf-xor instances,
it is to be expected that the more complex parity reasoning methods do not help
on such instances.
To evaluate whether this is the case,
we ran cryptominisat 2.9.2~\cite{SoosEtAl:SAT2009,Soos} on the 61 tree-like cnf-xor instances mentioned above in two modes:
(i) parity reasoning disabled with CNF input,
and
(i) parity reasoning enabled with cnf-xor form input and full Gaussian elimination.
The results in Fig.~\ref{Fig:Classification}(b) show that in this setting
it is beneficial to use CNF-level unit propagation instead of
the computationally more expensive Gaussian elimination method.

%
%
\subsection{Clausification of Tree-Like Parts}

As observed above,
a substantial number of real-world cnf-xor instances are not tree-like.
However, in many cases a large portion of the xor-clauses may appear in
tree-like parts of $\xorclauses$.
As an example,
consider the xor-clause conjunction $\xorclauses$
having the constraint graph shown in Fig.~\ref{Fig:TreeReduction}(a).
It is not \UP-deducible as ${\xorclauses \land a \land {\neg j}} \Models {e}$
but
${\xorclauses \land a \land {\neg j}} \UPDerivNot {e}$.
The xor-clauses $(i)$, $(g \X h \X i \X \T)$, $(e \X f \X g)$, and
$(d \X k \X m \X \T)$ form the tree-like part of $\xorclauses$.
\newcommand{\TreePart}[1]{\operatorname{treepart}(#1)}
Formally
the \emph{tree-like part of $\xorclauses$}, denoted by $\TreePart{\xorclauses}$,
can be defined recursively as follows:
(i) if there is a $\XC = (x_1 \X ... \X x_n \X \parity{})$ with $n \ge 1$ in $\xorclauses$ and an $n-1$-subset $W$ of $\Set{x_1,...,x_n}$ such that
each $x_i \in W$ appears only in $\XC$,
then $\TreePart{\xorclauses} = \Set{\XC}\cup\TreePart{\xorclauses\setminus\XC}$,
and
(ii) $\TreePart{\xorclauses} = \emptyset$ otherwise.

\begin{figure}[b]
  \centering
  \begin{tabular}{@{}c@{\qquad}c@{}}
    \includegraphics[scale=0.8]{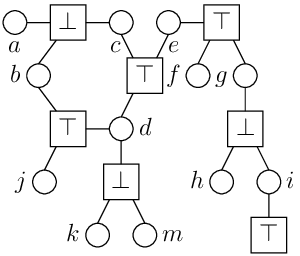}
    &
    \includegraphics[width=.5\textwidth]{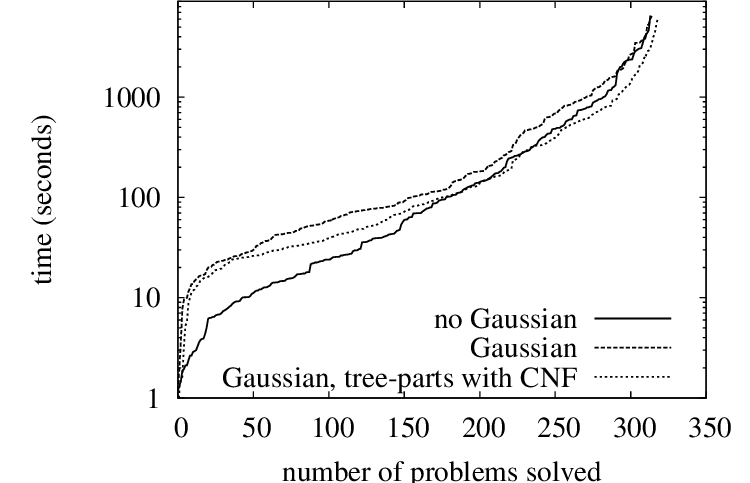}
    \\
    (a) & (b)
  \end{tabular}%
  \caption{A constraint graph (a), and run-times of cryptominisat on Hitag2 instances (b)}
  \label{Fig:TreeReduction}
\end{figure}

One can exploit such tree-like parts by applying only unit propagation on them
and
letting the more powerful but expensive xor-reasoning engines take care only of
the non-tree-like parts.
Sometimes such a strategy can lead to improvements in run time.
For example,
consider a set of 320 cnf-xor instances modeling known-plaintext attack on
Hitag2 cipher with 30--39 stream bits given.
These instances typically have 2600--3300 xor-clauses, of which roughly
one fourth are in the tree-like part.
Figure~\ref{Fig:TreeReduction}(b) shows the result when we run
cryptominisat 2.9.2~\cite{SoosEtAl:SAT2009,Soos} on
these instances with three configurations:
(i) Gaussian elimination disabled,
(ii) Gaussian elimination enabled,
and
(iii) Gaussian elimination enabled and the tree-like parts translated into CNF.
On these instances, applying the relatively costly Gaussian elimination
to non-tree-like parts only is clearly beneficial on the harder instances,
probably due to the fact that the Gaussian elimination matrices become smaller.
Smaller matrices consume less memory, are faster to manipulate, and can also give smaller xor-explanations for xor-implied literals.
On some other benchmark sets,
no improvements are obtained as
instances can contain no tree-like parts (e.g.~our instances modeling known-plaintext attack on \Trivium{} cipher)
or
the tree-like parts can be very small (e.g.~similar instances on the Grain cipher).
In addition, the effect is solver and xor-reasoning module dependent: we obtained no run time improvement with the solver of~\cite{LJN:ECAI2010}
applying equivalence reasoning.

We also ran the same cryptominisat configurations on all the
413 above mentioned non-tree-like SAT Competition benchmark instances.
The instances have a large number of xor-clauses (the largest number is 312707)
and
Fig.~\ref{Fig:SATCOMPTreeParts}(a) illustrates the relative tree-like part sizes.
As we can see,
a substantial amount of instances have a very significant tree-like part.
Figure~\ref{Fig:SATCOMPTreeParts}(b) shows the run-time results,
illustrating that applying Gaussian elimination on non-tree-like instances
can bring huge run-time improvements.
However,
one cannot unconditionally recommend using Gaussian elimination on
non-tree-like instances as on some instances,
especially in the ``industrial'' category,
the run-time penalty of Gaussian elimination was also huge.
Clausification of tree-like parts brings quite consistent improvement
in this setting.

\begin{figure}[tb]
  \centering
  \begin{tabular}{@{}c@{\quad}c@{}}
    \includegraphics[scale=0.465]{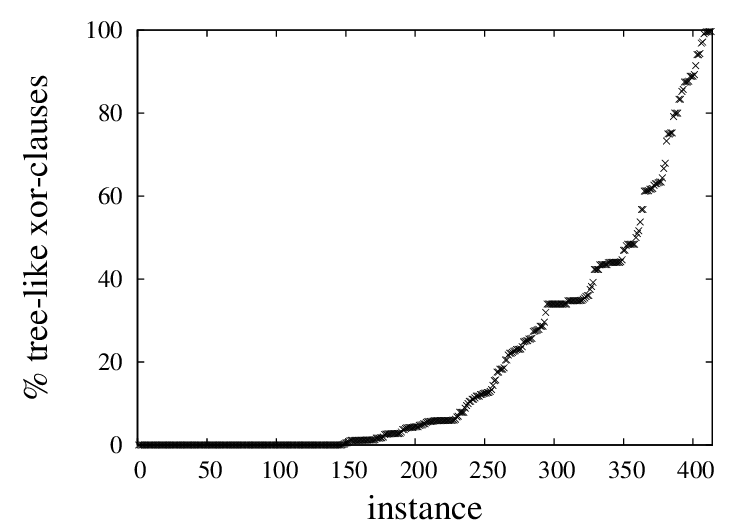} &
    \includegraphics[scale=0.465]{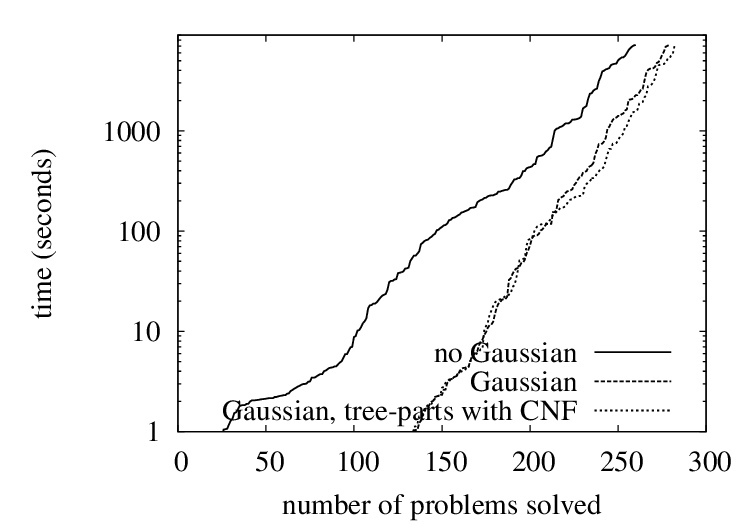}
    \\
    (a) & (b)
  \end{tabular}%
  \caption{Relative tree-like part sizes and run-times of non-tree-like instances.}
  \label{Fig:SATCOMPTreeParts}
\end{figure}

%% file: subst.tex
\section{Equivalence Reasoning}

\newcommand{\xuptrans}[1]{\ensuremath{\mathit{cycles}(#1)}}
\newcommand{\xorcycle}[3]{\ensuremath{\mathit{XC}(\Tuple{#1},\Tuple{#2},#3)}} 
\newcommand{\xorcycleS}[3]{\ensuremath{\mathit{XC}(#1,#2,#3)}} 
\newcommand{\VS}{\hspace{2mm}}
\newcommand{\VSS}{\hspace{2.5mm}}

As observed in the previous section,
unit propagation is not enough for deducing all xor-implied literals
on many practical cnf-xor instances.
We next perform a similar study for a stronger deduction system,
a form of equivalence reasoning~\cite{LJN:ECAI2010,LJN:ICTAI2011}.
Although it cannot deduce all xor-implied literals either,
on many problems it can deduce more
and
has been shown to be effective on some instance families.
The look-ahead based solvers \ToolName{EqSatz}~\cite{Li:AAAI2000} and
\ToolName{march\_eq}~\cite{HeuleEtAl:SAT2004} apply same kind of,
but not exactly the same,
equivalence reasoning we consider here.

To study equivalence reasoning on xor-clauses,
we introduce two equally powerful xor-reasoning systems:
``\SUBST{}''~\cite{LJN:ECAI2010} and ``\EC{}''~\cite{LJN:ICTAI2011}.
The first is simpler to implement and to present while
the second works here as a tool for analyzing
the structure of xor-clauses with respect to equivalence reasoning.
The ``\SUBST{}'' system simply adds two substitution rules to \UP{}:
\[
\begin{tabular}{cc@{\quad}c}
 \eqvruleP:\
 \inferencerule{x \oplus y \oplus \top}{C}{\simplification{C}{x}{y}}
 &
 \text{and}
 &
 \eqvruleN:\ 
 \inferencerule{x \oplus y}{C}{\simplification{C}{x}{y \oplus\top}}
\end{tabular}
\]
The ``\EC'' system,
standing for Equivalence Class based parity reasoning,
has the inference rules shown in Fig.~\ref{Fig:ECRules}.
\begin{figure}[b]
  \centering
  \small
  \begin{tabular}{@{}c@{\VSS}c@{\VSS}c@{}}
    $\inferencerule{x}{x \XX \True}{\False}$
    &
    $\inferencerule{x \XX \parity{1}}{x \XX y \XX \parity{2}}{y \xor (\parity{1} \XX \parity{2} \XX \top)}$
    &
    $\inferencerulen{{x_1 \XX x_2 \XX \parity{1} \XX \top}\VS\dots\VS{x_{n-1} \XX x_{n} \XX \parity{n-1} \XX \top}\VSS{x_1 \XX x_n \XX y \XX \parity{}}}
    {y \xor (\parity{1} \xor \parity{2} \xor \dots \xor \parity{n-1} \xor \parity{})} $
    \vspace{1mm}
    \\
    {\footnotesize (a) \UnConflictRuleP} &
    {\footnotesize (b) \BinUnRuleP} &
    {\footnotesize (d) \TerUnRuleP}
    \\
    \\
    \multicolumn{2}{@{}c@{}}{$\inferencerule{x \XX \parity{1}}{x \XX y \XX z \XX \parity{2}}{y \xor z \xor (\parity{1} \XX \parity{2} \XX \True)}$}
    &
    $\inferencerulenc{{x_1 \XX x_2 \XX \parity{1} \XX \top}\VS...\VS{x_{n-1} \XX x_n \XX \parity{n-1} \XX \top}\VSS{x_{n} \XX x_{1} \XX \parity{n} \XX \top}}
                {\bot}
                {\parity{1} \XX \dots \XX \parity{n} {=} \top}$
    \\
    \multicolumn{2}{@{}c@{}}{\footnotesize (c) \TerBinRuleP} &
    {\footnotesize (e) \ConflictRuleP}%
  \end{tabular}%
  \caption{Inference rules of \ECsys; the symbols $x,x_i,y,z$ are all variables while the $\parity{i}$ symbols are constants $\False$ or $\True$.}%
  \label{Fig:ECRules}%
\end{figure}
As there are no inference rules for xor-clauses with more than three variables,
longer xor-clauses have to be eliminated, e.g.,
by repeatedly applying the rewrite rule
$(x_1 \X x_2 \X \dots \X x_n) \leadsto {(x_1 \X x_2 \X y) \land (\neg y \X x_3 \X ... \X x_n)}$,
where $y$ is a fresh variable not occurring in other clauses.
We define \SUBST{}- and \EC-derivations, the relations $\SUBSTDeriv$ and $\ECDeriv$,
as well as \SUBST- and \EC-deducibility 
similarly to \UP-derivations, $\UPDeriv$, and \UP-deducibility, respectively.
\begin{example}
  Figure~\ref{Fig:Derivs} shows (parts of) \SUBST- and \EC-derivations
  from $\xorclauses \land (a) \land (\neg j)$,
  where $\xorclauses$ is the xor-clause conjunction shown in Fig.~\ref{Fig:TreeReduction}(a).
\end{example}
As shown in~\cite{LJN:ICTAI2011},
on cnf-xor instances with xor-clauses having at most three variables,
\SUBST{} and \EC{} can deduce exactly the same xor-implied literals.
Thus, such an instance $\xorclauses$ is
\SUBST-deducible if and only if it is \EC-deducible.

\begin{figure}[tb]
  \centering
  \begin{tabular}{c@{\qquad}c}
    \includegraphics[width=.45\textwidth]{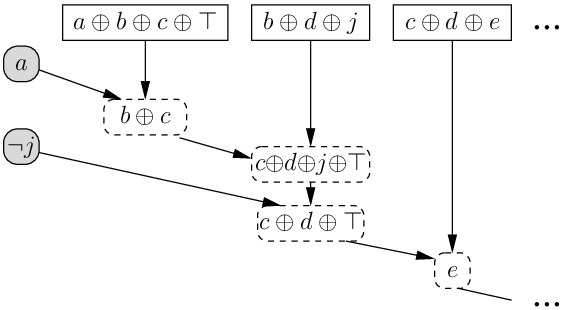}
    &
    \includegraphics[width=.45\textwidth]{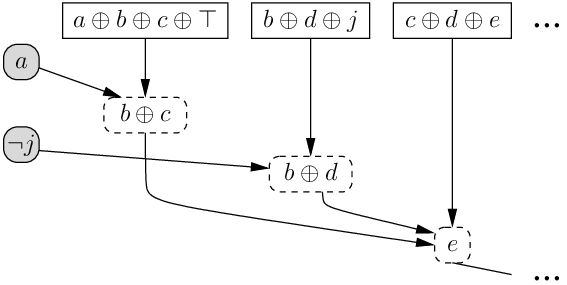}
    \\
    (a) a \SUBST-derivation
    &
    (b) an \EC-derivation
  \end{tabular}
  \caption{\SUBST- and \EC-derivations from $\xorclauses \land (a) \land (\neg j)$, where $\xorclauses$ is given in Fig.~\ref{Fig:TreeReduction}(a)}%
  \label{Fig:Derivs}%
\end{figure}

The \EC{}-system uses more complicated inference rules than \SUBST{},
but it allows us to characterize equivalence reasoning as
a structural property of constraint graphs.
The \EC{} rules $\UnConflictRule$,
$\BinUnRule$, and $\TerBinRule$  are for unit propagation on xor-clauses with
1--3 variables,
and the rules $\TerUnRule$ and $\ConflictRule$ for equivalence reasoning.
To simplify the following proofs and translations,
we consider a restricted class of xor-clauses.
A conjunction of xor-clauses $\xorclauses$ is in {\it 3-xor normal form} if
(i) every xor-clause in it has exactly three variables,
and
(ii) each pair of xor-clauses shares at most one variable.
Given a $\xorclauses$,
an equi-satisfiable 3-xor normal form formula can be obtained by
(i) eliminating unary and binary xor-clauses by unit propagation and substitution,
(ii) cutting longer xor-clauses as described above,
and
(iii) applying the following rewrite rule:
$(x_1 \X x_2 \X x_3) \land (x_2 \X x_3 \X x_4) \leadsto
(x_1 \X x_2 \X x_3) \land (x_1 \X x_4 \X \T)$.
In this normal form, $\ConflictRule$ is actually a shorthand
for two applications of $\TerUnRule$ and
one application of $\UnConflictRule$,
so the rule $\TerUnRule$ succinctly characterizes equivalence reasoning.
We now prove that the rule $\TerUnRule$ 
is closely related to the cycles in the constraint graphs.
An {\it xor-cycle} is an xor-clause conjunction of form
$(x_1 \XX x_2 \XX y_1 \equiv \parity{1}) \wedge\dots\wedge(x_{n-1} \XX  x_{n} \XX y_{n-1} \equiv \parity{n-1})\wedge(x_1 \XX x_n \XX y_n \equiv \parity{n})$,
abbreviated with $\xorcycle{x_1,...,x_n}{y_1,...,y_n}{\parity{}}$
where $\parity{} = \parity{1} \X ... \X \parity{n}$.
We call $x_1,...,x_n$ the {\it inner variables} and
$y_1,...,y_n$ the {\it outer variables} of the xor-cycle.
\begin{example}
  The cnf-xor instance shown in Fig.~\ref{Fig:TreeReduction}(a) has one xor-cycle
  $(a \X b \X c \X \T) \land (c \X d \X e) \land (b \X d \X j)$,
  where $b,c,d$ are the inner and
  $a,e,j$ the outer variables.
\end{example}
A key observation is that the $\TerUnRule$ rule can infer a literal
\emph{exactly when there is an xor-cycle}
with the values of the outer variables except for one already derived:
\begin{lemma}\label{Lem:XorImply}
  Assume an \EC-derivation $\ECderivation = \XC_1,\dots,\XC_n$ from
  $\psi = \xorclauses \wedge \AL_1 \wedge ... \wedge \AL_k $,
  where $\xorclauses$ is a 3-xor normal form xor-clause conjunction.
  There is an extension $\ECderivation'$ of $\ECderivation$
  where
  an xor-clause $ (y \equiv \parity{} \oplus \parity{1}' \oplus ... \oplus \parity{n-1}')$
   is derived using \TerUnRuleP{} on the xor-clauses
  $\{(x_1 \XX x_2 \equiv \parity{1} \XX \parity{1}')$, $...$, $(x_{n-1} \XX x_{n} \equiv \parity{n-1} \XX \parity{n-1}')$, $(x_1 \XX x_n \XX y \equiv \parity{n}) \} $
  if and only if
  there is an xor-cycle
  $\xorcycle{x_1,...,x_n}{y_1,...y_{n-1},y}{\parity{}} \subseteq
  \xorclauses$ where $\parity{} {=} \parity{1} \XX ... \XX \parity{n}$ such that for each $y_i \in \set{y_1,...,y_{n-1}}$ it
  holds that $ \psi \ECDeriv (y_i {\equiv} \parity{i}') $.
\end{lemma}

\input{substdetection2}

%
%
\subsection{Simulating equivalence reasoning with unit propagation}

The connection between equivalence reasoning and xor-cycles enables us to
consider a potentially more efficient way to implement equivalence reasoning.
We now present three translations that add redundant xor-clauses in the problem
with the aim that unit propagation is enough to always deduce all xor-implied
literals in the resulting xor-clause conjunction.
The first translation is based on the xor-cycles of the formula and
does not add auxiliary variables,
the second translation is based on explicitly communicating equivalences
between the variables of the original formula using auxiliary variables,
and
the third translation combines the first two.

The redundant xor-clause conjunction,
called an \emph{\EC{}-simulation formula} $\psi$,
added to $\xorclauses$ by a translation should satisfy the following:
(i) the satisfying truth assignments of $\xorclauses$ are exactly
the ones of $\xorclauses \land \psi$ when projected to
$\VarsOf{\xorclauses}$,
and
(ii) if $\IL$ is \EC-derivable from $\xorclauses \land (\AL_1) \land ... \land (\AL_k)$, then $\IL$ is \UP-derivable from $(\xorclauses \land \psi) \land (\AL_1) \land ...  \land (\AL_k)$.

%
\subsubsection{Simulation without extra variables.}

We first present an \EC{}-simulation formula for a given
3-xor normal form xor-clause conjunction $\xorclauses$
without introducing additional variables.
The translation adds one xor-clause with the all outer variables
per xor-cycle:
\[
 \xuptrans{\xorclauses} =
\bigwedge_{\xorcycle{x_1,...,x_n}{y_1,...,y_n}{\parity{}} \subseteq \xorclauses}(y_1 \X ...  \X y_n \equiv \parity{})
    \]
For example,
for the xor-clause conjunction $\xorclauses$ in Fig.~\ref{Fig:TreeReduction}(a)
$\xuptrans{\xorclauses} = (a \X e \X j \X \T)$.
Now $\xorclauses \land \xuptrans{\xorclauses} \land (a) \land (\neg j) \UPDeriv e$ although $\xorclauses \land (a) \land (\neg j) \UPDerivNot e$.
\begin{theorem}\label{Thm:XupTransCorrectness}
  If $\xorclauses$ is a 3-xor normal form xor-clause conjunction,
  then $\xuptrans{\xorclauses}$ is an \EC{}-simulation formula for $\xorclauses$.
\end{theorem}

\input{xuptransresults}

\input{expxupification}

\newcommand{\eijtrans}[1]{\ensuremath{\textup{Eq}}(#1)}
\newcommand{\opttransname}{\ensuremath{\textup{Eq}^\star}}
\newcommand{\eijtransname}{\ensuremath{\textup{Eq}}}
\newcommand{\opttrans}[1]{\ensuremath{\textup{Eq}^\star}(#1)}

\subsubsection{Simulation with extra variables: basic version.}

Our second translation $\eijtrans{\xorclauses}$ avoids the exponential increase
in size by introducing a quadratic number of auxiliary variables.
A new variable $e_{ij}$ is added for each pair of distinct variables
$x_i,x_j \in \VarsOf{\xorclauses}$,
with the intended meaning that $e_{ij}$ is true when $x_i$ and $x_j$ have the same value and false otherwise. 
We identify $e_{ji}$ with $e_{ij}$.
Now the translation is
\begin{eqnarray*}
  \eijtrans{\xorclauses} &=&
  \bigwedge_{(x_i \XX x_j \XX x_k {\equiv} \parity{}) \in \xorclauses}
  (e_{ij} \XX x_k \XX \T {\equiv} \parity{}) {\land}
  (e_{ik} \XX x_j \XX \T {\equiv} \parity{}) {\land}
  (x_i \XX e_{jk} \XX \T {\equiv} \parity{}) {\land}\\
  &&
  \bigwedge_{x_i,x_j,x_k \in \VarsOf{\xorclauses},i < j < k}(e_{ij} \X e_{jk} \X e_{ik} \equiv \top)
\end{eqnarray*}
where
(i) the first line ensures that if we can deduce that two variables in a ternary xor-clause are (in)equivalent, then we can deduce the value of the third variable,
and vice versa,
and
(ii) the second line encodes transitivity of (in)equivalences.
The translation enables unit propagation to deduce all \EC{}-derivable
literals over the variables in the original xor-clause conjunction:
%
%
\begin{theorem}
  \label{Thm:EijTransCorrectness}
  If $\xorclauses$ is an xor-clause conjunction in 3-xor normal form,
  then $\eijtrans{\xorclauses}$ is an \EC{}-simulation formula for $\xorclauses$.
  \label{Thm:Eij}
\end{theorem}

\input{opttrans}

\subsubsection{Strengthening equivalence reasoning by adding xor-clauses.}
Besides enabling unit propagation to simulate equivalence reasoning, the
translation $\opttrans{\xorclauses}$ has an another 
interesting property:
if $\xorclauses$ is not \SUBST{}-deducible, then 
even when $\xorclauses \wedge \AL_1 \wedge ... \wedge \AL_n \SUBSTDerivNot \IL $ for some xor-assumptions $\AL_1,...,\AL_n$, it might hold that
$ \xorclauses \wedge \opttrans{\xorclauses} \wedge \AL_1 \wedge ... \wedge \AL_n \SUBSTDeriv \IL $.
For instance,
let $\xorclauses$ be an xor-clause conjunction given in Fig.~\ref{Fig:NonSubstExamples}(b).
It holds that $ \xorclauses \wedge (x) \Models (z) $ but
$\xorclauses \wedge (x) \SUBSTDerivNot (z)$.
However, $\xorclauses \land \opttrans{\xorclauses} \wedge (x) \SUBSTDeriv (z)$;
the constraint graph of $\xorclauses \land \opttrans{\xorclauses}$ is
shown in Fig.~\ref{Fig:NonSubstExamples}(c).

\begin{figure}[t]
  \centering
  \begin{tabular}{c@{\quad}c@{\quad}c}
  \includegraphics[width=0.30\textwidth]{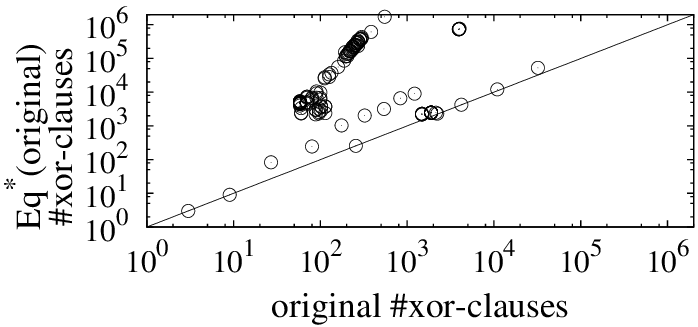} &
    \includegraphics[width=0.30\textwidth]{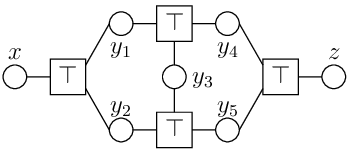} &
  \includegraphics[width=0.30\textwidth]{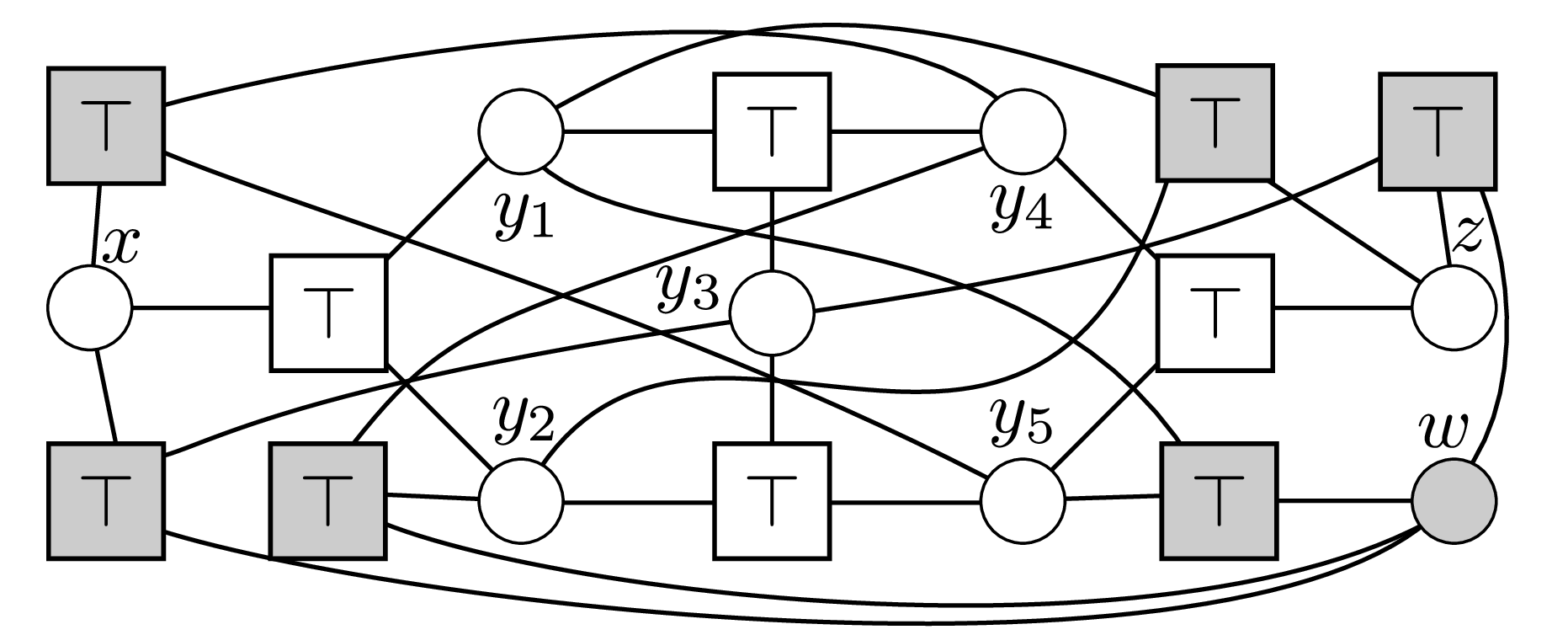}
    \\
    (a) & (b) & (c)
  \end{tabular}%
  \caption{(a) number of xor-clauses in our SAT'05 instances, (b) a non-\SUBST{}-deducible instance that becomes \SUBST{}-deducible with \opttransname in (c)}
  \label{Fig:NonSubstExamples}
\end{figure}

%% file: substdetection2.tex
\subsection{Detecting when equivalence reasoning is enough}

\newcommand{\twoColorInt}{\ensuremath{V_{\textup{in}}}}
\newcommand{\twoColorExt}{\ensuremath{V_{\textup{out}}}}
\newcommand{\twoColoring}{\ensuremath{\twoColorInt,\twoColorExt}}

The presence of xor-cycles in the problem implies that equivalence reasoning
might be useful, but does not give any indication of whether it is enough to
always deduce all xor-implied literals.
Again, we do not know any easy way to detect whether a given xor-clause conjunction is \SUBST{}-deducible (or equivalently, \EC-deducible).
However,
we can obtain a very fast structural test for approximating
\EC-deducibility as shown and analyzed in the following.

We say that a 3-xor normal form xor-clause conjunction $\xorclauses$ is
{\it cycle-partitionable} if
there is a
partitioning $(\twoColorInt{}, \twoColorExt{})$ of $\VarsOf{\xorclauses}$
such that
for each xor-cycle
$\xorcycleS{X}{Y}{\parity{}}$ in $\xorclauses$
it holds that $X \subseteq \twoColorInt{}$ and $Y \subseteq \twoColorExt{}$.
That is,
there should be no variable that appears as an inner variable in one xor-cycle
and
as an outer variable in another. 
For example,
the instance in Fig.~\ref{Fig:TreeReduction}(a) is cycle-partitionable
as $(\{b,c,d\},\{a,e,f,...,m\})$ is a valid cycle-partition.
On the other hand,
the one in Fig.~\ref{Fig:UPCG}(c) is not cycle-partitionable
(although it is \UP-deducible and thus \EC-deducible).
If such cycle-partition can be found,
then equivalence reasoning is enough to always deduce all xor-implied literals.
\begin{theorem}\label{Thm:ColoringTheorem}
  If a 3-xor normal form xor-clause conjunction $\xorclauses$
  is cycle-partitionable,
  then it is \SUBST{}-deducible (and thus also \EC-deducible).
\end{theorem}
Detecting whether a cycle-partitioning exists can be efficiently implemented
with a variant of Tarjan's algorithm for strongly connected components.

To evaluate the accuracy of the technique,
we applied it to the SAT Competition instances discussed 
in Sect.~\ref{Sect:TreeClassification}.
The results are shown in the ``cycle-partitionable'' and ``probably \SUBST'' columns in
Fig.~\ref{Fig:Classification}(a),
where the latter gives the number of instances for which our random testing procedure described in Sect.~\ref{Sect:TreeClassification} was not able to show
that the instance is not \SUBST-deducible.
We see that the accuracy of the cycle-partitioning test is (probably) not
perfect in practice although for some instance families it works very well.

%% file: xuptransresults.tex
The translation is intuitively suitable for problems that have a small number
of xor-cycles, such as the DES cipher. Each instance of our DES benchmark (4
rounds, 2 blocks) has 28--32 xor-cycles.
We evaluated experimentally the translation on this benchmark using cryptominisat 2.9.2, minisat 2.0, minisat 2.2, and minisat 2.0 extended with the \UP{} xor-reasoning module.
The benchmark set has 51 instances and the clauses of each instance are permuted 21 times randomly to negate the effect of propagation order.
The results are shown in
Fig.~\ref{Fig:Simulation}(a).
The translation manages to slightly reduce solving time for cryptominisat, but this does not happen for other solver configurations based on minisat, so the slightly improved performance is not completely due to simulation of equivalence reasoning using unit propagation.
The xor-part (320 xor-clauses of which 192 tree-like) in DES is negligible compared to cnf-part (over 28000 clauses), so a great
reduction in solving time is not expected.

\begin{figure}[t]
  \centering
  \begin{tabular}{@{}c@{\quad}c@{}}
    \includegraphics[width=.49\textwidth]{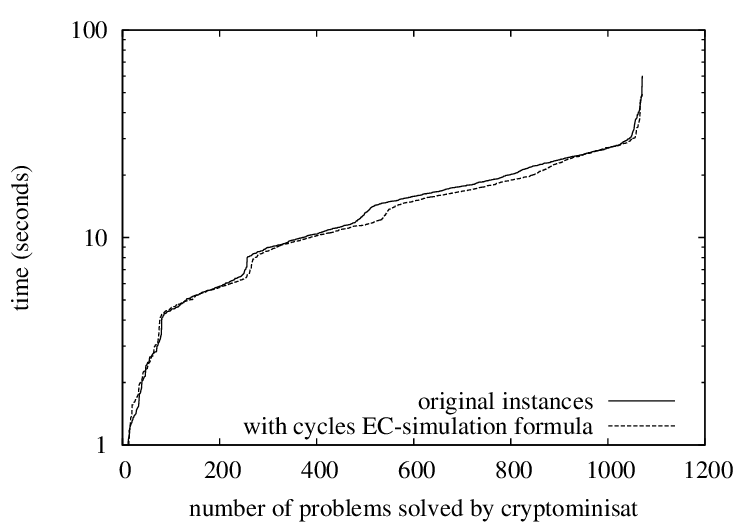}
    &
    \includegraphics[width=.47\textwidth]{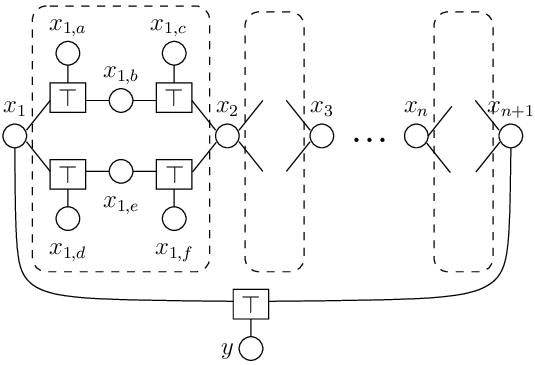}
    \\
    (a) & (b)%
  \end{tabular}%
  \vspace{-2mm}
  \caption{The $\xuptrans{\xorclauses}$ translation on DES instances (a), and the constraint graph of $\diamondcycles$ (b).}
  \label{Fig:Simulation}
\end{figure}

%% file: expxupification.tex
Although equivalence reasoning can be simulated with unit propagation by
adding an xor-clause for each xor-cycle,
this is not feasible for all instances in practice due to the large number of xor-cycles.
We now prove that,
\emph{without} using auxiliary variables,
there are in fact families of xor-clause conjunctions for which
all \EC{}-simulation formulas, whether based on xor-cycles or not,
are exponential.
%
%
Consider the xor-clause conjunction
$\diamondcycles =
 (x_1 \X x_{n+1} \X y) \land
 \bigwedge_{i=1}^n (x_i \XX x_{i,a} \XX x_{i,b}) 
   \land (x_{i,b} \XX x_{i,c} \XX x_{i+1})
   \land (x_i \XX x_{i,d} \XX x_{i,e})
   \land (x_{i,e} \XX x_{i,f} \XX x_{i+1})$
whose constraint graph is shown in
Fig.~\ref{Fig:Simulation}(b).
Observe that $\diamondcycles$ is cycle-partitionable and
thus \SUBST/\EC-deducible.
But all its \EC-simulation formulas are at least of exponential size if no auxiliary variables are allowed:%
%
\begin{lemma}\label{Lem:DiamondCycles}
  Any \EC{}-simulation formula $\psi$ for $\diamondcycles$
  with $\VarsOf{\psi} = \VarsOf{\diamondcycles}$
  contains at least $2^n$ xor-clauses.
\end{lemma}


%% file: opttrans.tex
\subsubsection{Simulation with extra variables: optimized version.}

The translation $\eijtrans{\xorclauses}$ adds a cubic number of clauses
with respect to the variables in $\xorclauses$.
This is infeasible for many real-world instances.
The third translation combines the first two translations by
implicitly taking into account the xor-cycles
in $\xorclauses$
while adding auxiliary variables where needed.
The translation $\opttrans{\xorclauses}$ is presented in Fig.~\ref{Fig:OptEij}.
The xor-clauses added by $\opttrans{\xorclauses}$ are a subset of
$\eijtrans{\xorclauses}$ and
the meaning of the variable $e_{ij}$ remains the same.
The intuition behind the translation, on the level of constraint graphs,
is to iteratively shorten xor-cycles by ``eliminating'' one variable
at a time by adding auxiliary variables that
``bridge'' possible equivalences over the eliminated variable.
The line 2 in the pseudo-code picks a variable $x_j$ to eliminate.
While the correctness of the translation does not depend on the choice,
it is sensible to pick a variable that shares xor-clauses with
fewest variables because
the number of xor-clauses produced in lines 3--9 is then smallest.
The loop in line 3 iterates over all possible xor-cycles where the selected
variable $x_j$ and two ``neighboring'' non-eliminated variables
$x_i$,$x_k$ may occur as inner variables.
The line 4 checks if there already is an xor-clause that has both $x_i$ and $x_k$.
If so, then in line 5 an existing variable is used as $e_{ik}$
capturing the equivalence between the variables $x_i$ and $x_k$.
If the variable $\parity{ik}$ is $\top$, then $e_{ik}$ is true when the variables $x_i$ and $x_k$ have the same value. 
The line 9 adds an xor-clause ensuring that transitivity of equivalences
between the variables $x_i$, $x_j$,and $x_k$
can be handled by unit propagation.

\newcommand{\vleft}{\ensuremath{V}}
\begin{figure}[bt]
{\small
\begin{tabbing}
99.x\={mm}\={mm}\={mm}\=\kill
$\opttrans{\xorclauses}$:\; start with $ \xorclauses' = \xorclauses $ and $V = \VarsOf{\xorclauses}$ \\
1.\>\While{} ($V \not = \emptyset$):\\
2.\>\>$x_j$ $\leftarrow$ extract a variable $v$ from $V$ minimizing $|\VarsOf{\Setdef{C\in \xorclauses'}{v \in \VarsOf{C}}} \cap V|$\\
3.\>\>\For{} each 
$(x_i \XX x_j \XX e_{ij} {\equiv} \parity{ij}){,}(x_j \XX x_k \XX e_{jk} {\equiv}\parity{jk}) {\in} \xorclauses' \mbox{ such that}\,x_i{,}x_k {\in} V {\wedge} x_i {\not =} x_j {\not =} x_k$\\
4.\>\>\>\If{} $ (x_i \oplus x_k \oplus y \equiv \parity{ik}') \in \xorclauses'$ \\
5.\>\>\>\>$e_{ik} \leftarrow y$; $\parity{ik} \leftarrow \parity{ik}'$ \\
6.\>\>\>\Else{}\\
7.\>\>\>\>$e_{ik} \leftarrow \mbox{new variable}$; $\parity{ik} \leftarrow \top $ \\
8.\>\>\>\>$\xorclauses' \leftarrow \xorclauses' \wedge (x_i \oplus x_k \oplus e_{ik} \equiv \parity{ik})$\\
9.\>\>\>$\xorclauses' \leftarrow \xorclauses' \wedge (e_{ij} \oplus e_{jk} \oplus e_{ik} \equiv \parity{ij} \oplus \parity{jk} \oplus \parity{ik})$ \\
10.\> \Return{} $\xorclauses' \backslash \xorclauses$%
\end{tabbing}%
}%
\vspace{-2mm}
\caption{The $\opttransname$ translation}
\label{Fig:OptEij}
\end{figure}%

\begin{example}
  Consider the xor-clause conjunction
  $\xorclauses = (x_1 \XX x_2 \XX x_4) \land (x_2 \XX x_3 \XX x_5) \land
   (x_5 \XX x_7 \XX x_8) \land (x_4 \XX x_6 \XX x_7)$
  shown in Fig.~\ref{Fig:OptTransExample}(a).
  The translation $\opttrans{\xorclauses}$ first selects one-by-one
  the variables in $\set{x_1,x_3,x_6,x_8}$ as each appears in
  only one xor-clause.
  The loop in lines 3--9 is not executed for any of them.
  The remaining variables are $V {=} \set{x_2,x_4,x_5,x_7}$.
  Assume that $x_2$ is selected.
  The loop in lines 3--9 is entered with values $x_i{=}x_4$, $x_j {=} x_2$, 
  $e_{ij} {=} x_1$, $x_k {=} x_5$, $e_{jk} {=} x_3$,
  $\parity{ij} {=} \top$, and $\parity{jk} {=} \top$.
  The condition in line 4 fails, so
  the xor-clauses $ (x_4 \XX x_5 \XX e_{45} {\equiv} \top) $ and $ (x_1 \XX x_3 \XX e_{45} {\equiv} \top)$,where $e_{45}$ is a new variable, are added.
  The resulting instance is shown in Fig.~\ref{Fig:OptTransExample}(b).
  Assume that $x_5$ is selected.
  The loop in lines 3--9 is entered with values $ x_i {=} x_4$, $x_j {=} x_5$, 
  $e_{ij} {=} e_{45}$, $x_k {=} x_7$, $e_{jk} {=} x_8$,
  $\parity{ij} {=} \top$, and $\parity{jk} {=} \top$.
  The condition in line 4 is true, so $e_{ik}{=}x_6 $, and the xor-clause $ (x_6 \XX x_8 \XX e_{45} {\equiv} \top) $ is added in line 9.
  The final result is shown in Fig.~\ref{Fig:OptTransExample}(c).
\end{example}

\begin{figure}[t]
  \centering
  \begin{tabular}{c@{\qquad}c@{\qquad}c}
    \includegraphics[width=0.26\textwidth]{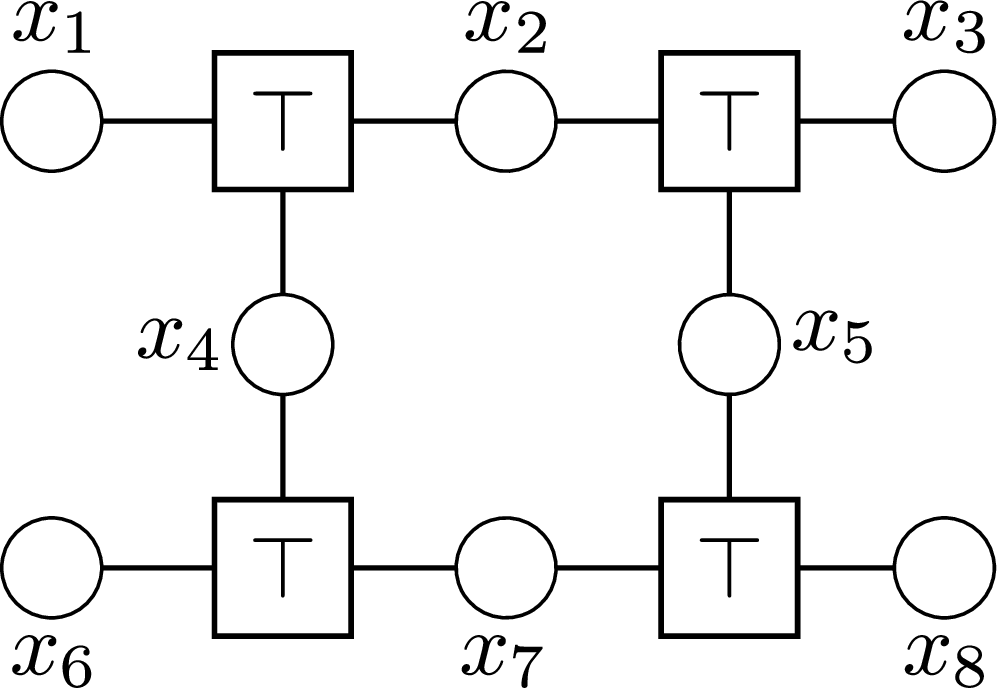} &
    \includegraphics[width=0.26\textwidth]{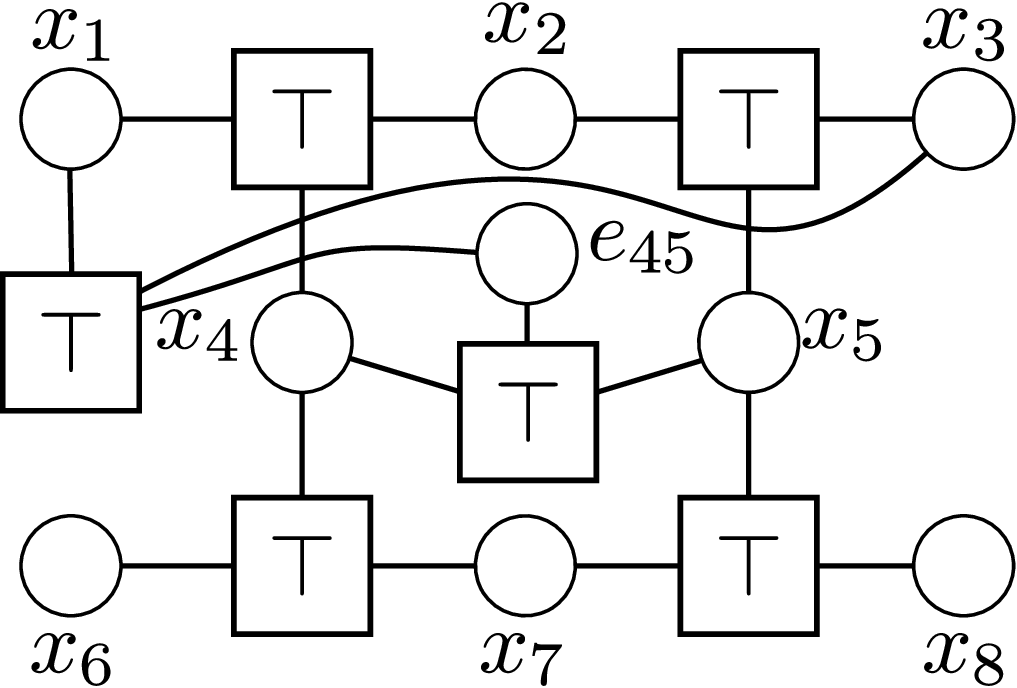} &
    \includegraphics[width=0.26\textwidth]{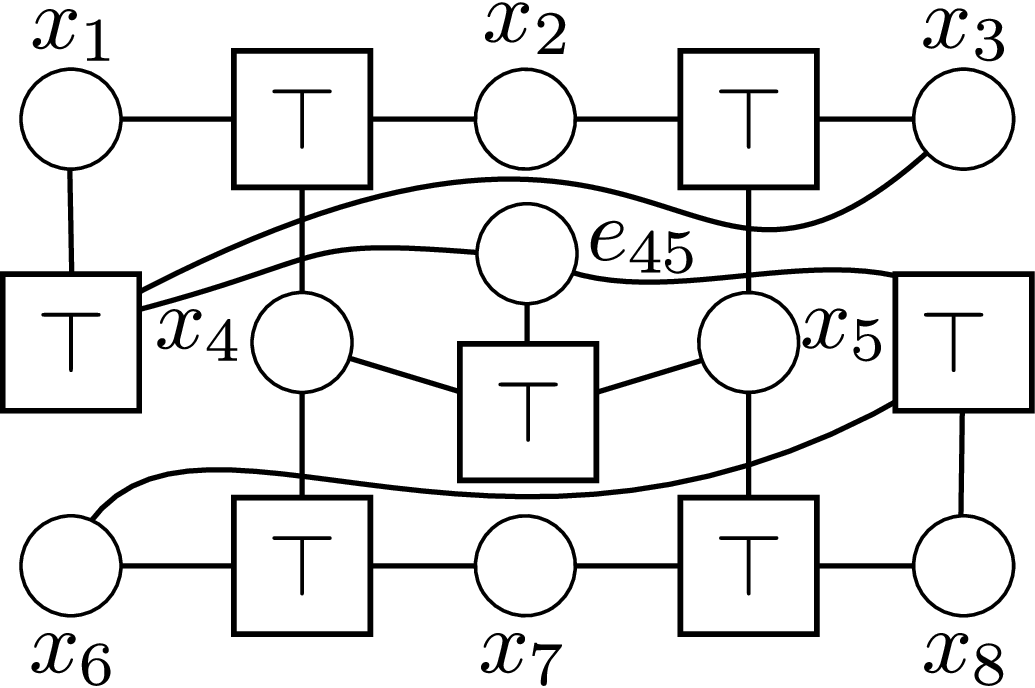}
    \\
    (a) & (b) & (c)
  \end{tabular}%
  \vspace{-2mm}
  \caption{Constraint graphs illustrating how the translation $\opttransname$ adds new xor-clauses}
  \label{Fig:OptTransExample}
\end{figure}

\begin{theorem}\label{Thm:OptTransCorrectness}
  If $\xorclauses$ is an xor-clause conjunction in 3-xor normal form,
  then $\opttrans{\xorclauses}$ is an \EC{}-simulation formula for $\xorclauses$.
\end{theorem}
\input{translationsizes}

\subsubsection{Experimental evaluation.}
To evaluate the translation $\opttransname{}$,
we ran cryptominisat 2.9.2, and glucose 2.0 (SAT Competition 2011 application track winner)
on the 123 SAT 2005 Competition cnf-xor instances preprocessed into 3-xor normal form
with and without $\opttransname{}$.
The results are shown in
Fig.~\ref{Fig:OptTransResults}. The number of
decisions is greatly reduced, and this is reflected in solving time on many
instances. Time spent computing $\opttransname{}$ is measured in seconds and is negligible compared to solving time.
On some instances, the translation adds a very
large number of xor-clauses (as shown in Fig.~\ref{Fig:NonSubstExamples}a) and the computational overhead of simulating
equivalence reasoning using unit propagation becomes prohibitively large.
For highly ``xor-intensive''
instances it is probably better to use more powerful parity reasoning; cryptominisat 2.9.2
with Gaussian elimination enabled solves majority of these instances in a few
seconds. A hybrid approach first trying if $\opttransname{}$ adds a
moderate number of xor-clauses, and if not, resorting to stronger 
parity reasoning could, thus, be an effective technique for solving cnf-xor
instances.

\begin{figure}[t]
  \centering
  \includegraphics[width=\textwidth]{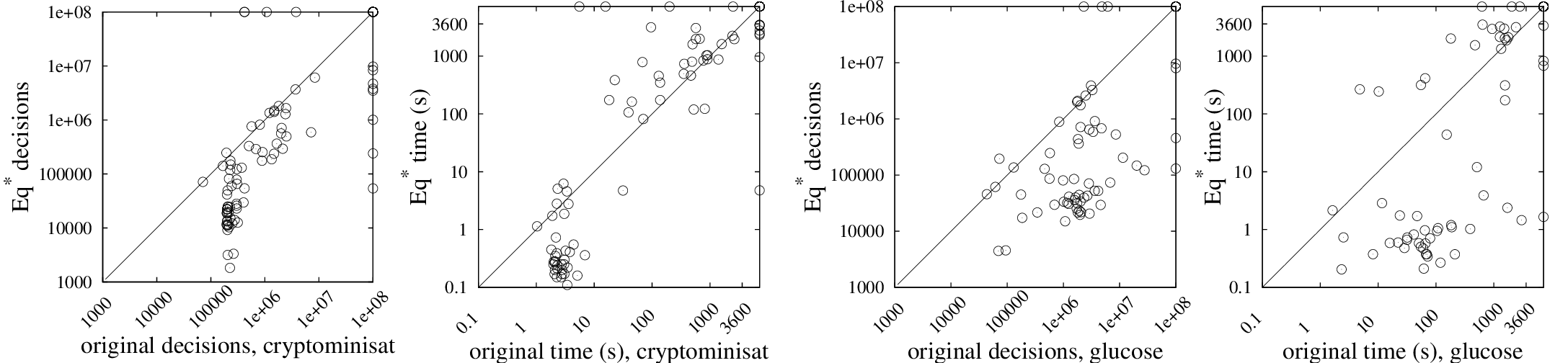}
  \vspace{-6mm}
  \caption{Experimental results with/without $\opttransname{}$ (cryptominisat on the left, glucose on the right)}
  \label{Fig:OptTransResults}
\end{figure}

%% file: translationsizes.tex
The translation $\opttransname{}$ usually adds fewer variables and
xor-clauses than $\eijtransname{}$.
Fig.~\ref{Fig:TranslationSizes} shows a comparison of
the translation sizes on four cipher benchmarks.
The translation $\eijtransname{}$ yields an impractically large
increase in formula size,
while the translation $\opttransname{}$ adds still a
manageable
number of new variables and xor-clauses.


\begin{figure}[t]
\centering
  \begin{scriptsize}
\begin{tabular}{|l|c|c|c|c|c|c|c|}
\hline
\textbf{Benchmark}
&
\multicolumn{2}{|c|}{
    $\phi = \mbox{Original}$
}
&
\multicolumn{2}{|c|}{
   $\phi \wedge \eijtrans{\phi}$ 
}
&
\multicolumn{2}{|c|}{
    $\phi \wedge \opttrans{\phi}$ 
}
&
\\
&
\multicolumn{1}{|c}{
  vars
}
&
\multicolumn{1}{c|}{
  xor-clauses
}
&
\multicolumn{1}{|c}{
  vars
}
&
\multicolumn{1}{c|}{
  xor-clauses
}
&
\multicolumn{1}{|c}{
  vars
}
&
\multicolumn{1}{c|}{
  xor-clauses
}
& $\frac{|\phi \wedge \eijtrans{\phi}|}{
         |\phi \wedge \opttrans{\phi}|}$
\\
\hline
DES (4 rounds 2 blocks) 
& 3781
& 320
& $7 \times 10^6$
& $9 \times 10^9$
& 3813
& 416
& $ 2.2 \times 10^7$
\\

Grain (16 bit)
& 9240
& 6611
& $43 \times 10^6$
& $131 \times 10^9$
& 71670 
& 3957571
& 33212.0
\\

Hitag2 (33 bit) 
& 6010
& 3747
& $18 \times 10^6$
& $36 \times 10^9$
& 21092
& 338267
& 106904.4
\\

\Trivium{} (16 bit)
& 11485
& 8591
& $66 \times 10^6$
& $252 \times 10^9$
& 351392 
& 30588957 
& 8252.1
\\
\hline
\end{tabular}%
  \end{scriptsize}%
\caption{Comparison of the translation sizes for $\eijtransname{}$
and $\opttransname{}$ on cipher benchmarks}
\label{Fig:TranslationSizes}
\end{figure}

%% file: conclusions.tex
\section{Conclusions}

We have given efficient approximating tests for detecting whether unit
propagation or equivalence reasoning is enough to achieve full propagation in a
given parity constraint set.
To our knowledge the computational complexity of exact versions of these
tests is an open problem; they are certainly in co-NP but are they in P?

We have shown that equivalence reasoning can be simulated with unit propagation by adding a polynomial amount of redundant parity constraints to the problem.
We have also proven that without introducing new variables,
an exponential number of new parity constraints would be needed
in the worst case.
We have found many real-world problems for which unit propagation or
equivalence reasoning achieves full propagation.
The experimental evaluation of
the presented translations suggests that equivalence reasoning can be
efficiently simulated by unit propagation.

\subsubsection*{Acknowledgments.}
This work has been financially supported by the Academy  of Finland under the Finnish Centre of Excellence in Computational Inference (COIN)
    and Hecse Helsinki Doctoral Programme in Computer Science.

%% file: proofs.tex
\newpage

\newcommand{\XCB}{E}
\newcommand{\Equal}{\equiv}

\section{Proofs}
For two xor-clauses
$\XC = (x_1 \X ... \X x_k \Equal p)$
and
$\XCB = (y_1 \X ... \X y_l \Equal q)$,
we define their linear combination xor-clause by
$\XC \X \XCB =
 (x_1 \X ... \X x_k \X y_1 \X ... \X y_l \Equal p \X q)$.
Some fundamental, easy to verify properties of xor-clauses are
${\XC \land \XCB} \Models {\XC \X \XCB}$,
${\XC \land \XCB} \Models {\XC \land (\XC \X \XCB)}$,
and
${\XC \land (\XC \X \XCB)} \Models {\XC \land \XCB}$.


\newenvironment{relemma}[1]{\renewcommand{\thelemma}{#1}\begin{lemma}}{\end{lemma}}
\newenvironment{retheorem}[1]{\renewcommand{\thetheorem}{#1}\begin{theorem}}{\end{theorem}}

%

\subsection{Proof of Theorem~\ref{Theorem:Trees}}

\begin{retheorem}{\ref{Theorem:Trees}}
  If a conjunction of xor-clauses $\xorclauses$ is tree-like,
  then it is \UP-deducible.
\end{retheorem}
\begin{proof}
  %
  Assume that the constraint graph of $\xorclauses$ is a tree;
  the case when it is a union of trees follows straightforwardly.
  Proof by induction on the number of xor-clauses in $\xorclauses$.

  Base cases.
  (i) If $\xorclauses$ is the empty conjunction,
  then it is both tree-like and \UP-deducible.
  (ii) If $\xorclauses$ consists of a single xor-clause $\XC$,
  then it is both tree-like and \UP-deducible.


  Induction hypothesis.
  The lemma holds for all tree-like conjunctions that have at most $n$
  xor-clauses.

  Induction step.
  Take any tree-like xor-clause conjunction $\xorclauses$ with $n+1$ xor-clauses
  and
  any xor-clause $\XC$ in it.
  Let $\xorclauses'$ denote the xor-clause conjunction obtained from $\xorclauses$ by removing $\XC$
  and
  let $\xorclausesI{1}',\ldots,\xorclausesI{q}'$ be the variable-disjoint
  xor-clause clusters of $\xorclauses'$.
  Each $\xorclausesI{i}'$ is obviously tree-like,
  and
  $\XC$ includes exactly one variable $x'_i$ occurring in $\xorclausesI{i}'$.
  Let $Y = \Set{y_1,\ldots,y_m}$ be the set of variables that occur in $\XC$ but
  not in $\xorclauses'$.
  Each model of $\xorclauses'$ is a disjoint union of models of
  $\xorclausesI{1}',\ldots,\xorclausesI{q}'$.
  Take any $\IL,\AL_1,...,\AL_k \in \LitsOf{\xorclauses}$ such that
  ${\xorclauses \land \AL_1 \land ... \land \AL_k} \Models \IL$;
  the case when $\xorclauses \land \AL_1 \land ... \land \AL_k$ is unsatisfiable
  can be proven similarly.
  \begin{enumerate}
  \item
    If ${\VarsOf{\AL_1 \land \ldots \land \AL_k} \cap Y} \subset Y$,
    then the models of $\xorclauses \land \AL_1 \land ... \land \AL_k$,
    when projected to $\VarsOf{\xorclauses} \setminus Y$,
    are the ones of $\xorclauses' \land \AL_1 \land ... \land \AL_k$ as
    the xor-clause $\XC$ can be satisfied in each by letting
    the variable(s) in $Y \setminus \VarsOf{\AL_1 \land \ldots \land \AL_k}$ 
    take appropriate values.
    Now there are two cases to consider:
    \begin{enumerate}
    \item
      If $\VarsOf{\IL} \not \subseteq Y$ but $\VarsOf{\IL} \subseteq \VarsOf{\xorclausesI{i}'}$
      for some $i \in \Set{1,...,q}$,
      then ${\xorclauses \land \AL_1 \land ... \land \AL_k} \UPDeriv \IL$ holds
      because $\xorclausesI{i}'$ is tree-like and \UP-deducible by the induction
      hypothesis.
    \item
      If $\VarsOf{\IL} \subseteq Y$,
      then it must be that case that
      $\VarsOf{\AL_1 \land \ldots \land \AL_k} \cap Y = {Y \setminus \VarsOf{\IL}}$
      and
      ${\xorclauses \land \AL_1 \land ... \land \AL_k} \Models x'_i \X  \parity{i}$
      for each $i \in \Set{1,...,q}$ and some $\parity{i} \in \Set{\F,\T}$.
      Now for each $i \in \Set{1,...,q}$ it holds that
      ${\xorclauses \land \AL_1 \land ... \land \AL_k} \UPDeriv x'_i \X \parity{i}$
      because $\xorclausesI{i}'$ is tree-like and
      \UP-deducible by the induction hypothesis.
      Thus ${\xorclauses \land \AL_1 \land ... \land \AL_k} \UPDeriv \IL$ holds.
    \end{enumerate}
  \item
    If ${\VarsOf{\AL_1 \land \ldots \land \AL_k} \cap Y} = Y$,
    then assume, without loss of generality,
    that the variable of $\IL$ occurs in the sub-tree $\xorclausesI{q}$.
    If $\xorclausesI{q} \land \AL_1 \land \ldots \land \AL_k \Models \IL$,
    then $\xorclausesI{q} \land \AL_1 \land \ldots \land \AL_k \UPDeriv \IL$
    by induction hypothesis
    and
    thus $\xorclauses \land \AL_1 \land \ldots \land \AL_k \UPDeriv \IL$.
    If $\xorclausesI{q} \land \AL_1 \land \ldots \land \AL_k \ModelsNot \IL$,
    then it must be that 
    $\xorclausesI{i} \land \AL_1 \land \ldots \land \AL_k \Models x'_i \X \parity{i}$,
    and thus $\xorclausesI{i} \land \AL_1 \land \ldots \land \AL_k \UPDeriv x'_i \X \parity{i}$ by induction hypothesis,
    for each $i \in \Set{1,...,q-1}$ and some $\parity{i} \in \Set{\F,\T}$.
    After this unit propagation can derive $x'_q \X \parity{q}$
    for a $\parity{q}\in\Set{\F,\T}$ and
    then $\xorclausesI{q} \land \AL_1 \land \ldots \land \AL_k \land (x'_q \X \parity{q}) \Models \IL$
    and thus 
    $\xorclausesI{q} \land \AL_1 \land \ldots \land \AL_k \land (x'_q \X \parity{q}) \UPDeriv \IL$.
  \end{enumerate}%
\qed
\end{proof}

%

\subsection{Proof of Lemma~\ref{Lem:XorImply}}

\begin{relemma}{\ref{Lem:XorImply}}
  Assume an \EC-derivation $\ECderivation$ from
  $\psi = \xorclauses \wedge \AL_1 \wedge ... \wedge \AL_k $,
  where $\xorclauses$ is a 3-xor normal form xor-clause conjunction.
  There is an extension $\ECderivation'$ of $\ECderivation$
  where 
  an xor-clause $ (y \equiv \parity{} \oplus \parity{1}' \oplus ... \oplus \parity{n-1}')$
   is derived using \TerUnRuleP{} on the xor-clauses
  $\{(x_1 \XX x_2 \equiv \parity{1} \XX \parity{1}')$, $...$, $(x_{n-1} \XX x_{n} \equiv \parity{n-1} \XX \parity{n-1}')$, $(x_1 \XX x_n \XX y \equiv \parity{n}) \} $
  if and only if
  there is an xor-cycle
  $\xorcycle{x_1,...,x_n}{y_1,...y_{n-1},y}{\parity{}} \subseteq
  \xorclauses$ where $\parity{} {=} \parity{1} \XX ... \XX \parity{n}$ such that for each $y_i \in \set{y_1,...,y_{n-1}}$ it
  holds that $ \psi \ECDeriv (y_i {\equiv} \parity{i}') $.
\end{relemma}
\begin{proof}
Assume that there is an extension $\ECderivation'$ of $\ECderivation$
where 
an xor-clause $ (y \equiv 
\parity{} \oplus \parity{1}' \oplus ... \oplus \parity{n-1}')$
is derived using
 \TerUnRuleP{} on the xor-clauses in $\phi
=\{(x_1 \XX x_2 \equiv \parity{1} \XX \parity{1}')$, $...$, $(x_{n-1} \XX x_{n} \equiv \parity{n-1} \XX \parity{n-1}')$, $(x_1 \XX x_n \XX y \equiv \parity{n}) \} $.
Since $\xorclauses$ is in 3-xor normal form,
each xor-clause $ (x_i \XX x_j \equiv \parity{i} \XX \parity{i}') \in \phi $ 
is derived from the conjunction $(x_i \XX x_j \XX y_i \equiv \parity{i})
    \wedge (y_i \equiv \parity{i}') $, so 
both of these xor-clauses must be in
$\ECderivation'{}$. This implies that for each variable $y_i$, it holds $
\xorclauses \ECDeriv (y_i \equiv \parity{i}')$. Also, the xor-clauses $(x_1
\XX x_2 \XX y_1 \equiv \parity{1} )$, $(x_2 \XX x_3 \XX y_2 \equiv
\parity{2} )$, ..., $(x_{n-1} \XX x_n \XX y_{n-1} \equiv
        \parity{n-1} )$ must be in $\xorclauses$.
Thus,  the conjunction $\xorclauses$ has an xor-cycle
$\xorcycle{x_1,...,x_n}{y_1,...,y_{n-1},y}{\parity{}}$.

Assume now that there  there is an xor-cycle
$\xorcycle{x_1,...,x_n}{y_1,...,y_{n-1},y}{\parity{}}$ in $\xorclauses$ and
for each variable $y_i \in \set{y_1,...,y_{n-1}}$ it holds $ \psi \ECDeriv
(y_i \equiv \parity{i}')$. 
An extension $\ECderivation'$ to $\ECderivation$ 
such that the xor-clause $ (y \equiv \parity{} \oplus \parity{1}' \oplus ... \oplus \parity{n-1}')$ is derived using $\TerUnRuleP{}$ on the xor-clauses
in $\{ (x_1 \XX x_2 \equiv \parity{1} \XX \parity{1}')$, $...$, $(x_{n-1} \XX x_{n} \equiv \parity{n-1} \XX \parity{n-1}')$, $(x_1 \XX x_n \XX y \equiv \parity{n}) \} $
can be constructed as follows:
\begin{enumerate}
\item Add each xor-clause in the xor-cycle $\xorcycle{x_1,...,x_n}{y_1,...,y_{n-1},y}{\parity{}}$
to $\ECderivation{}$. 
\item For each $y_i \in
\set{y_1,...,y_{n-1}}$, add a number of derivation steps including the xor-clause $ (y_i \equiv \parity{i}')$ to $\ECderivation{}$ because
$\psi \ECDeriv (y_i \equiv \parity{i}')$, 
\item Apply
\TerBinRuleP{} on pairs of xor-clauses $ (x_1 \XX x_2 \XX y_1 {\equiv} \parity{1})
{\wedge} (y_1 {\equiv} \parity{1}')$, $(x_2 \XX x_3 \XX y_2 \equiv \parity{2}) \wedge (y_2 \equiv \parity{2})$, $...$,
$(x_{n-1} \XX x_n \XX y_{n-1} \equiv \parity{n-1}) \wedge (y_{n-1} \equiv \parity{n-1}')$
and thus adding xor-clauses $(x_1 \XX x_2 \equiv \parity{1} \XX \parity{1}')$,
    $\dots$, $(x_{n-1} \XX x_{n} \equiv \parity{n-1} \XX \parity{n-1}')$ to
    \ECderivation{}. 
\item All the premises for \TerUnRuleP{} are in place and we can derive $ (y
\equiv \parity{} \oplus \parity{1}' \oplus ... \oplus \parity{n-1}')$ using \TerUnRuleP{} on $\phi$.
\end{enumerate}
\qed
\end{proof}

%
\newcommand{\BigLinComb}{\sum}
\subsection{Proof of Theorem~\ref{Thm:ColoringTheorem}}
\begin{lemma}[from \cite{LJN:ICTAI2012full}]\label{Lem:LinearCombination}
Let $\psi$ be a conjunction of xor-constraints (xor-clauses).
Now $\psi$ is unsatisfiable if and only if
there is a subset $S$ of xor-constraints (xor-clauses) in $\psi$ such that
$\BigLinComb_{\XC \in S} \XC = (\F \Equal \T)$.
If $\psi$ is satisfiable and $\XCB$ is an xor-constraint (xor-clause),
then $\psi \Models \XCB$ 
if and only if
there is a subset $S$ of xor-constraints (xor-clauses) in $\psi$ such that
$\BigLinComb_{\XC \in S} \XC = \XCB$.
\end{lemma}

\begin{retheorem}{\ref{Thm:ColoringTheorem}}
  If a 3-xor normal form xor-clause conjunction $\xorclauses$
  is cycle-partitionable,
  then it is \SUBST{}-deducible (and thus also \EC-deducible).
\end{retheorem}
\begin{proof}
Let $\xorclauses$ be a cycle-partitionable conjunction in 3-xor normal form. 
We assume that $\xorclauses \wedge \AL_1 \wedge \dots \wedge \AL_k $ is
satisfiable. The case when $\xorclauses \wedge \AL_1 \wedge \dots \wedge \AL_k$
can be proven similarly.
Assume that $\xorclauses \wedge \AL_1 \wedge \dots
\dots \wedge \AL_n \Models \IL$.  
By Lemma~\ref{Lem:LinearCombination}, there is a subset $S$ of xor-clauses in
$\xorclauses$ such that $\BigLinComb_{\XC \in S} \XC = \IL$.
Since $\xorclauses$ is cycle-partitionable, it clearly holds
that $\xorclauses'$ is cycle-partitionable also. Let \twoColoring{} be a cycle-partitioning for $\xorclauses'$.
 The proof proceeds by case analysis on the structure of the constraint graph
of $\xorclauses'$. Because $\xorclauses'$ is cycle-partitionable,
the constraint graph of $\xorclauses'$ does not have any cycles involving the variables in $\twoColorExt{}$.
This means that we can partition the conjunction $\xorclauses'$ into a sequence
of pairwise disjoint conjunctions of xor-clauses $\phi'_1,\phi'_2,...,\phi'_k$
such that the constraint graph of each conjunction $ \phi'_i$ is a connected
component, $\xorclauses' = \phi'_1 \cup ... \cup \phi'_k $, and for all
distinct pairs $\phi'_i,\phi'_j$ it holds that $|\VarsOf{\phi'_i} \cap
\VarsOf{\phi'_j}| \leq 1$ and $|\VarsOf{\phi'_i} \cap \VarsOf{\phi'_j}|
\subseteq \twoColorExt{}$. 
\begin{enumerate}
\item If it holds that 
$ (\VarsOf{C_1} \cup ... \cup \VarsOf{C_m}) \cap \twoColorExt{}
\subseteq \VarsOf{C}$
    it suffices to consider any conjunction $\phi'_i$
for which $\VarsOf{l} \subseteq \VarsOf{\phi'_i}$ holds, because 
$\VarsOf{\phi'_i} \cap \twoColorExt{} \subseteq \VarsOf{\AL_1,...,\AL_k,\IL}$,
    and thus $\phi'_i \wedge \AL_1 \wedge ... \wedge \AL_n \Models \IL $.   
We consider
    the cases:
\begin{enumerate}
\item If the constraint graph of $\phi'_i$ is tree-like, then by Theorem~\ref{Theorem:Trees} it holds that $\xorclauses \wedge \AL_1 \wedge ... \wedge \AL_n \SUBSTDeriv \IL$.
\item Otherwise, the constraint graph of $\phi'_i$ is not tree-like, and has at least one xor-cycle.
Due to the cycle-partitioning and the presence of at least one xor-cycle, the conjunction $\phi'_i$ can be 
partitioned into a finite set of partially overlapping xor-cycles $\phi'_i =
\xorcycleS{X_1}{Y_1}{\parity{1}} \cup ... \cup \xorcycleS{X_l}{Y_l}{\parity{l}}
$. By definition, for each xor-cycle $\xorcycleS{X_i}{Y_i}{\parity{i}}$ it
holds that each variable $v \in X_i$ has exactly two occurrences in the
xor-cycle $\xorcycleS{X_i}{Y_i}{\parity{i}}$. Let $\twoColorInt{}' = \VarsOf{C}
\cap \twoColorInt{} \cap \VarsOf{\phi'_i}$. If $\twoColorInt{}'\not =
\emptyset$, then there exists a variable $x \in \twoColorInt{}'$ with three
occurrences in $\phi'_i$ in the xor-clauses $C_a = (x \oplus x_a \oplus y_a
\equiv \parity{a}) $, $C_b = (x \oplus x_b \oplus y_b \equiv
\parity{b})$, and $C_c = (x \oplus x_c \oplus y_c \equiv
\parity{c})$ because $\xorclauses$ is in 3-xor normal form. There are
two xor-cycles $\xorcycleS{X_i}{Y_i}{\parity{i}}$ and $\xorcycleS{X_j}{Y_j}{\parity{j}}$ such
that $ C_a $ and $C_b$ are in $\xorcycleS{X_i}{Y_i}{\parity{i}}$ and $C_c$ is in
$\xorcycleS{X_j}{Y_j}{\parity{j}}$. The xor-cycles $\xorcycleS{X_i}{Y_i}{\parity{i}}$
and $\xorcycleS{X_j}{Y_j}{\parity{j}}$ overlap, so there is another variable $x' \in \twoColorInt{}'$ also with three occurrences in $\phi'_i$ in the xor-clauses
 $C_a' = (x' \oplus x_a' \oplus y_a'
\equiv \parity{a}') $, $C_b' = (x' \oplus x_b' \oplus y_b' \equiv
\parity{b}')$, and $C_c' = (x' \oplus x_c' \oplus y_c' \equiv
\parity{c}')$ such that $C_a'$ and $C_b'$ are in $\XC{X_i,Y_i,\parity{i}}$ and
$C_c'$ is in $\XC{X_j,Y_j,\parity{j}}$. Thus, the number of inner variables of $\phi'_i$ in the linear combination is even, that is, $|\VarsOf{C} \cap \twoColorInt{} \cap \VarsOf{\phi'_i}| \equiv 0 \mod 2 $.

We consider two cases:
\begin{enumerate}
\item 
If $\VarsOf{\IL} \in \twoColorInt{}$, then the intersection
$\VarsOf{\AL_1,...,\AL_n} \cap \twoColorInt{}$ is non-empty, because $|\VarsOf{C} \cap
\twoColorInt{} \cap \VarsOf{\phi'_i}| \equiv 0 \mod 2 $. It follows that there is an xor-cycle $\xorcycleS{X}{Y}{\parity{}}$ such that $\VarsOf{\IL} \in X$, and
$\VarsOf{\AL_1,...,\AL_n} \cap X$ is non-empty. We can construct a
\SUBST{}-derivation $\SUBSTderivation{}$ from $\xorcycleS{X}{Y}{\parity{}} \wedge
\AL_1 \wedge ...  \wedge \AL_n $ by repeatedly applying $\unitruleP{}$
or $\unitruleN{}$
to all xor-clauses $C'$ in $\xorcycleS{X}{Y}{\parity{}}$ such that
$\VarsOf{C'} \cap \VarsOf{\AL_1,...,\AL_n} \cap \twoColorExt{} \not =
\emptyset$. Now there are $|X|$ xor-clauses of the form
$ (x_i \oplus x_j \oplus \parity{i})$ in the \SUBST{}-derivation $\SUBSTderivation$. Because
$\VarsOf{\AL_1,...\AL_n} \cap X$ is non-empty, we can add the xor-clauses
$ (\AL_1), ..., (\AL_n)$ to the \SUBST{}-derivation and then continue applying
\unitruleP{} and \unitruleN{} until $\IL$ is derived. It follows that
$\xorclauses \wedge \AL_1 \wedge ... \wedge \AL_n \SUBSTDeriv \IL$.
 
\item Otherwise, $\VarsOf{\IL} \in \twoColorExt{}$ and there is an xor-cycle
$\xorcycleS{X}{Y}{\parity{}}$ such that $\VarsOf{\IL} \in Y$. It holds that $ Y
\backslash{\VarsOf{\IL}} \subseteq \VarsOf{\AL_1,\dots,\AL_n}$, so by
Lemma~\ref{Lem:XorImply} it holds that 
$\xorclauses \wedge \AL_1 \wedge ...
\wedge \AL_n \ECDeriv \IL$ and thus also
$\xorclauses \wedge \AL_1 \wedge ...
\wedge \AL_n \SUBSTDeriv \IL$.
\end{enumerate}
\end{enumerate}
\item Otherwise, it holds that 
    $ (\VarsOf{C_1} \cup ... \cup \VarsOf{C_m}) \cap \twoColorExt{}
\not \subseteq \VarsOf{C}$.
Then there must be at least one conjunction
$\phi'_i \in \set{\phi'_1, ...,\phi'_k}$ such that $(\VarsOf{\phi'_i} \cap \twoColorExt{})
\backslash \VarsOf{\AL_1,...,\AL_n} = \set{y}$ for some variable $y$ . By a
similar reasoning as above we can prove that
$\phi'_i \wedge \AL_1 \wedge ... \wedge \AL_n \SUBSTDeriv (y \oplus \parity{y})$. This can be applied repeatedly until it has
been proven for some conjunction $\phi'_j$ such that  $\VarsOf{\IL} \in
\VarsOf{\phi'_j}$ and for each variable $v \in (\VarsOf{\phi'_j} \cap \twoColorExt{}
        \backslash \VarsOf{\IL})$ it holds that $\xorclauses \wedge \AL_1
\wedge ... \wedge \AL_n \SUBSTDeriv (v \oplus \parity{v})$. Then, again by similar
reasoning as above we can prove that $\xorclauses \wedge \AL_1
\wedge ... \wedge \AL_n \SUBSTDeriv \IL $.
\end{enumerate}
\qed
\end{proof}

%

\subsection{Proof of Theorem~\ref{Thm:XupTransCorrectness}}

\begin{retheorem}{\ref{Thm:XupTransCorrectness}}
  If $\xorclauses$ is a 3-xor normal form xor-clause conjunction,
  then $\xuptrans{\xorclauses}$ is an \EC{}-simulation formula for $\xorclauses$.
\end{retheorem}
\begin{proof}
%
We first prove that
the satisfying truth assignments of $\xorclauses$ are exactly
the ones of $\xorclauses \land \xuptrans{\xorclauses}$ when projected to
to $\VarsOf{\xorclauses}$. It holds by definition that 
$\xorclauses \land \xuptrans{\xorclauses} {\Models} \xorclauses$, so
it suffices to show that $ \xorclauses {\Models} \xuptrans{\xorclauses}$. 
Each xor-clause $(y_1 \oplus ... \oplus y_n {\equiv} \parity{}) {\in} \xuptrans{\xorclauses}$ corresponds to an 
xor-cycle $ \xorcycle{x_1,...,x_n}{y_1,...,y_n}{\parity{}}$, that is,
a conjunction of xor-clauses $(x_1 \XX x_2 \XX y_1 {\equiv} \parity{1}) \wedge\dots\wedge(x_{n-1} \XX  x_{n} \XX y_{n-1} \equiv \parity{n-1})\wedge(x_1 \XX x_n \XX y_n \equiv \parity{n})
\subseteq \xorclauses$ where $\parity{} = \parity{1} \oplus ... \oplus \parity{n}$. Observe that $ (y_1 \oplus ... \oplus y_n \equiv \parity{})$ is a 
linear combination of the xor-clauses in $\xorcycle{x_1,...,x_n}{y_1,...,y_n}{\parity{}}$, so it holds that
$\xorcycle{x_1,...,x_n}{y_1,...,y_n}{\parity{}} \Models (y_1 \oplus ... \oplus y_n \equiv \parity{})$ and $\xorclauses \Models \xuptrans{\xorclauses}$.
   
We now prove that if $\IL$ is \EC-derivable from $\xorclauses \land (\AL_1)
   \land ... \land (\AL_k)$, then $\IL$ is \UP-derivable from $(\xorclauses
            \land \xuptrans{\xorclauses}) \land (\AL_1) \land ...  \land
    (\AL_k)$.
Assume that a literal $\IL$ is \EC-derivable
from $ \psi = \xorclauses \land (\AL_1) \land ... \land (\AL_k)$. 
This implies
that there is an \EC-derivation \ECderivation{} from $\psi$ and the literal
$\IL$ is derived from the xor-clauses $C_1, ..., C_n$ in \ECderivation{}
using one of the inference rules of \EC.
We prove by structural induction that
$\IL$ is \UP-derivable from $\psi' = \psi \land \xuptrans{\xorclauses} \land (\AL_1) \land ... \land (\AL_k)$. The induction hypothesis is that there is a
\UP-derivation \UPderivation{} from $\psi'$ such that the xor-clauses
$C_1,...,C_n$ are in \UPderivation{}.  The inference rules
\UnConflictRuleP{},  \TerBinRuleP{}, and \BinUnRuleP{} are special cases of
\unitruleP{} and \unitruleN{}, and \ConflictRuleP{} can be simulated by \TerUnRule{} and \UnConflictRuleP{}, 
    so it suffices to show that the inference rule
\TerUnRule{} can be simulated with \unitruleP{} and
\unitruleN{}. In the case that $\IL = (y \equiv \parity{} \oplus \parity{1}' \oplus ... \oplus \parity{n-1}') $ is derived using
the inference rule \TerUnRule{}, by Lemma~\ref{Lem:XorImply} there must be an
xor-cycle $\xorcycle{x_1,...,x_n}{y_1,...,y_{n-1},y}{\parity{}}$ in
$\xorclauses$
such that for each $ y_i \in \set{y_1,...,y_{n-1}}$ it holds that
$\psi \ECDeriv (y_i \equiv \parity{i}')$. By induction hypothesis
it holds that $\psi' \UPDeriv (y_i \equiv \parity{i}') $ for each $ y_i \in
\set{y_1,...,y_{n-1}}$.  The xor-clause $ (y_1 \XX ... \XX y_{n-1} \XX y \equiv
        \parity{})$ is in $\xuptrans{\xorclauses}$, so $\psi' \UPDeriv (y
           \equiv \parity{} \oplus \parity{1}' \oplus ... \oplus
            \parity{n-1}')$. It follows that $\IL$ is
UP-derivable from $(\xorclauses \wedge \xuptrans{\xorclauses}) \land (\AL_1) \land ... \land (\AL_k)$.
\qed
\end{proof}

%

\subsection{Proof of Lemma~\ref{Lem:DiamondCycles}}

\begin{relemma}{\ref{Lem:DiamondCycles}}
  Any \EC{}-simulation formula $\psi$ for $\diamondcycles$
  with $\VarsOf{\psi} = \VarsOf{\diamondcycles}$
  contains at least $2^n$ xor-clauses.
\end{relemma}
\begin{proof}
Let $\psi$ to be an \EC{}-simulation formula for $\diamondcycles$.
  Observe the constraint graph of $\diamondcycles$ in
  Fig.~\ref{Fig:Simulation}(b).
  %
  There are two ways to traverse each ``diamond gadget'' that connects the
  variables $x_i$ and $x_{i+1}$, so there are $2^n$ xor-cycles of the form
  $\xorcycleS{X}{Y}{\top}$ where $X =
  \Tuple{x_1,x_1',x_2,x_2',...,x_n',x_{n+1},x_1}$, $x_i' \in \set{x_{i,b},
    x_{i,e}}$, $Y=\Tuple{y_1,y_1',...,y_n,y_n',y}$, $\Tuple{y_i,y_i'} {\in}
  \set{\Tuple{x_{i,a},x_{i,c}}, \Tuple{x_{i,d},x_{i,f}}}$. 
  By Lemma~\ref{Lem:XorImply}, for each such xor-cycle $\xorcycleS{X}{Y}{\top}$
  there is an EC-derivation from
  $\diamondcycles \wedge (Y \backslash \set{y})$
  where $y$ can be added using \TerUnRuleP{} on the xor-clauses
  $\{(x_1 \XX x_2 \equiv \top)$,
      $(x_2 \XX x_3 \equiv \top)$,
      $...$,
      $(x_n \XX x_{n+1} \equiv \top)$,
      $(x_1 \XX x_{n+1} \XX y \equiv \top)\}$.
Now, let $\xorcycleS{X}{Y}{\top}$ be any such xor-cycle.
Let $ Y_a = Y \backslash \set{y}$. 
 Note that for each variable $x \in \VarsOf{\diamondcycles} \backslash Y$ it holds that
    $\diamondcycles \wedge Y_a \not \Models (x)$
    and 
    $ \diamondcycles \wedge Y_a \not \Models (x \oplus \top)$, so for each xor-clause $C$ in $\psi$ 
    such that $ \VarsOf{C} \backslash Y \not = \emptyset$, it holds 
    that $ \diamondcycles \wedge Y_a \wedge C \UPDerivNot (x)$ and
    $ \diamondcycles \wedge Y_a \wedge C \UPDerivNot (x \oplus \top)$, so xor-clauses in $\psi$
    that have other variables than the variables in $Y$ cannot help in deriving the xor-clause
    $(y)$. Also, there cannot be an xor-clause $C$ in $\psi$ such that $|\VarsOf{C}| < |Y|$, $y \in \VarsOf{C}$, 
      and $ \VarsOf{C} \backslash Y = \emptyset$, because then for some literal
      $\IL$ and for some xor-assumptions $\AL_1,...,\AL_k$ it would hold that $
      \diamondcycles \wedge \psi \wedge \AL_1 \wedge ... \wedge \AL_k \Models \IL $ but 
      $ \diamondcycles \wedge \AL_1 \wedge ... \wedge \AL_k \not \Models  \IL
      $, and thus $\psi$ would not be an \EC{}-simulation formula for
      $\diamondcycles$.
Clearly, there is exactly one xor-clause $C = \bigoplus Y $ such that $y \in C,
\VarsOf{C} = Y$ such that $\diamondcycles \wedge (Y \backslash \set{y})
\wedge C \UPDeriv (y)$. For each xor-cycle $\xorcycleS{X}{Y}{\top}$ it holds that the corresponding
xor-clause $C$ must be in the \EC{}-simulation formula $\psi$, so $\psi$
must have at least $2^n$ xor-clauses.  
\qed
\end{proof}

%

\subsection{Proof of Theorem~\ref{Thm:EijTransCorrectness}}

\begin{retheorem}{\ref{Thm:EijTransCorrectness}}
  If $\xorclauses$ is an xor-clause conjunction in 3-xor normal form,
  then $\eijtrans{\xorclauses}$ is an \EC{}-simulation formula for $\xorclauses$.
\end{retheorem}
\begin{proof}
We first prove that
the satisfying truth assignments of $\xorclauses$ are exactly
the ones of $\xorclauses \land \eijtrans{\xorclauses}$ when projected to
$\VarsOf{\xorclauses}$. It holds by definition that 
$\xorclauses \land \eijtrans{\xorclauses} \Models \xorclauses$, so
it suffices to show that if $\TA$ is a satisfying truth assignment
for $ \xorclauses$, it can be extended to a satisfying truth assignment $\TA'$
for $\eijtrans{\xorclauses}$. Assume that $\TA$ is a truth assignment
such that $ \TA \Models \xorclauses$. 
Let $\TA' $ be a truth assignment identical to $\TA$ except for
the following additions. Let $ x_i,x_j $ be any two distinct variables in $\VarsOf{\xorclauses}$.
The conjunction $\eijtrans{\xorclauses}$ has a corresponding variable $ e_{ij}$. 
If $ \TA \Models (x_i \oplus x_j \oplus \top)$, then add $ e_{ij} $ to $\TA'$.
Otherwise add $ \neg e_{ij}$ to $\TA'$.
For each xor-clause $(x_i \X x_j \X x_k \equiv \parity{}) \in \xorclauses$,
the conjunction  $\eijtrans{\xorclauses}$ contains three xor-clauses 
  $  (e_{ij} \XX x_k \XX \T \equiv \parity{}) $,
  $ (e_{ik} \XX x_j \XX \T \equiv \parity{}) $, and
  $(x_i \XX e_{jk} \XX \T \equiv \parity{}) $.
It holds that $ \TA' \Models e_{ij} \leftrightarrow (x_i \oplus x_j \oplus \top)$,
        $ \TA' \Models e_{jk} \leftrightarrow (x_j \oplus x_k \oplus \top)$,
            and $ \TA' \Models e_{ik} \leftrightarrow (x_i \oplus x_k \oplus \top)$. By substituting $e_{ij}$ with $(x_i \oplus x_j \oplus \top)$, we get
$ \TA' \Models (e_{ij} \XX x_k \XX \top \equiv \parity{}) \leftrightarrow
  ((x_i \oplus x_j \oplus \top) \XX x_k \XX \top \equiv \parity{})$.
  By simplifying this equivalence we get 
$ \TA' \Models (e_{ij} \XX x_k \XX \top \equiv \parity{}) \leftrightarrow
  (x_i \oplus x_j \XX x_k \equiv \parity{})$, and since
  $ \TA' \Models  (x_i \oplus x_j \XX x_k \equiv \parity{})$,
  it follows that $\TA' \Models (e_{ij} \XX x_k \XX \top \equiv \parity{})$.
 The reasoning for the other two xor-clauses is analogous, so
    it also holds that 
  $ \TA' \Models (e_{ik} \XX x_j \XX \T \equiv \parity{}) $, and
  $ \TA' \Models (x_i \XX e_{jk} \XX \T \equiv \parity{}) $.
  The conjunction $\eijtrans{\xorclauses}$ has also an xor-clause 
$(e_{ij} \X e_{jk} \X e_{ik} \equiv \top)$
  for
  each distinct triple $x_i,x_j,x_k \in \VarsOf{\xorclauses}$.
  Since $ \TA' \Models e_{ij} \leftrightarrow (x_i \oplus x_j \oplus \top)$,
        $ \TA' \Models e_{jk} \leftrightarrow (x_j \oplus x_k \oplus \top)$,
            and $ \TA' \Models e_{ik} \leftrightarrow (x_i \oplus x_k \oplus \top)$,
                it holds that $ \TA' \Models (e_{ij} \X e_{jk} \X e_{ik} \equiv \top)
                \leftrightarrow ((x_i \oplus x_j \oplus \top)
                                \X (x_j \oplus x_k \oplus \top)
                                \X (x_i \oplus x_k \oplus \top))$.
                By simplifying the equivalence we get,
        $ \TA' \Models (e_{ij} \X e_{jk} \X e_{ik} \equiv \top) \leftrightarrow (\top)$,
  and further $ \TA' \Models (e_{ij} \X e_{jk} \X e_{ik} \equiv \top) $.

We now prove that if $\IL$ is \EC-derivable from $\xorclauses \land \AL_1
   \land ... \land \AL_k$, then $\IL$ is \UP-derivable from $\xorclauses
            \land \eijtrans{\xorclauses} \land \AL_1 \land ...  \land
    \AL_k$.
 Assume that a literal $\IL$ is \EC-derivable from $ \psi = \xorclauses \land
 \AL_1 \land ... \land \AL_k$.  This implies that there is an
 \EC-derivation \ECderivation{} from $\psi$ and the literal $\IL$ is derived
 from the xor-clauses $C_1, ..., C_n$ in \ECderivation{} using one of the
 inference rules of \EC.  We prove by structural induction that $\IL$ is
 \UP-derivable from $\psi' = \xorclauses \land \eijtrans{\xorclauses} \land \AL_1 \land ...
 \land \AL_k$. The induction hypothesis is that there is a \UP-derivation
 \UPderivation{} from $\psi'$ such that the xor-clauses $C_1,...,C_n$ are in
 \UPderivation{}.  The inference rules \UnConflictRuleP{},  \TerBinRuleP{}, and
 \BinUnRuleP{} are special cases of \unitruleP{} and \unitruleN{}, 
 and \ConflictRuleP{} can be simulated by \TerUnRule{} and \UnConflictRuleP{}, 
 so it
 suffices to show that the inference rule \TerUnRule{} 
 can be simulated with \unitruleP{} and \unitruleN{}. In the case that $\IL =
 (y \equiv \parity{}\oplus \parity{1}' \oplus ... \oplus \parity{n-1}')$ is derived using the inference rule \TerUnRule{}, by
 Lemma~\ref{Lem:XorImply} there must be an xor-cycle
 $\xorcycle{x_1,...,x_n}{y_1,...,y_{n-1},y}{\parity{}}$ in $\xorclauses$ 
 where $ \parity{} = \parity{1} \oplus ... \oplus \parity{n}$  such
 that for each $ y_i \in \set{y_1,...,y_{n-1}}$ it holds that $\psi \ECDeriv
 (y_i \equiv \parity{i}')$. By induction hypothesis it holds that $\psi' \UPDeriv (y_i \equiv \parity{i}')$ for each $ y_i 
\in \set{y_1,...,y_{n-1}}$. It follows that for each xor-clause $C \in
 \{(e_{12} \oplus \parity{1} \oplus \parity{1}' \oplus \top), (e_{23} \oplus \parity{2} \oplus \parity{2}' \oplus \top), ..., (e_{{n-1}{n}} \oplus
         \parity{n-1} \oplus \parity{n-1}' \oplus \top)\}$ it holds that $ \psi' \UPDeriv C $.  The conjunction
 $\eijtrans{\xorclauses}$ has an xor-clause $ (e_{ij} \XX e_{jk} \XX e_{ik} \equiv \top)$
 for each triple of distinct variables $x_i,x_j,x_k \in \VarsOf{\xorclauses}$,
 so by repeatedly applying \unitruleP{} and \unitruleN{} we can derive $
    (e_{1n} \equiv \parity{} \oplus \parity{1}' \oplus ... \oplus
     \parity{n-1}')$, and then $ (y \equiv \parity{} \oplus \parity{1}' \oplus
         ... \oplus \parity{n-1}') $.
By induction it follows that $\IL$ is UP-derivable from $\xorclauses \land \eijtrans{\xorclauses} \land \AL_1
    \land ... \land \AL_k$.
\qed
\end{proof}
%

%

\subsection{Proof of Theorem~\ref{Thm:OptTransCorrectness}}

\begin{lemma}\label{Lem:CycleCompress}
Given a conjunction of xor-clauses $\xorclauses$ in 3-xor normal form ,
an elimination order $\Tuple{x_1,\dots,x_n}$ for the algorithm
$\opttransname$,
  and an xor-cycle
$\xorcycleS{X}{Y}{\parity{}}$, $|X| \geq 3$ in $\xorclauses \wedge \opttrans{\xorclauses}$ where
$X \subseteq \VarsOf{\xorclauses}$, there are variables $x \in X$, $y',y''
\in Y$ such that 
$ (y \oplus y' \oplus y'' \equiv \parity{}')$ is in 
    $\xorclauses \wedge \opttrans{\xorclauses} $ 
    and if $|X| > 3$, the conjunction $\xorclauses \wedge \opttrans{\xorclauses}$ has an xor-cycle $ \xorcycleS{X\backslash \set{x}}{(Y \cup \set{y}) \backslash \set{y',y''}}{\parity{}' \oplus \parity{}}$.
\end{lemma}
\begin{proof}
Assume that $\opttrans{\xorclauses}$ has an xor-cycle
$\xorcycleS{X}{Y}{\parity{}}$, $|X| \geq 3$ where $X \subseteq
\VarsOf{\xorclauses}$. 
Let $x_j$ be the first variable in the elimination order sequence $\Tuple{x_1,\dots,x_n}$ such that $x_j$ is also in $X$.
There are two different variables $x_i,x_k \in X$ such that $ (x_i \oplus x_j
\oplus e_{ij} \equiv \parity{ij}), (x_j \oplus x_k \oplus e_{jk} \equiv
\parity{jk})$
are in $\xorcycleS{X}{Y}{\parity{}}$.
The algorithm in Fig.~\ref{Fig:OptEij}
enters line 3 and $x_i,x_j,x_k \in X, e_{ij},e_{jk} \in Y$ hold. After
this, the conjunction $\xorclauses \wedge \opttrans{\xorclauses}$ contains the xor-clauses $
(e_{ij} \oplus e_{jk} \oplus e_{ik} \equiv \parity{ij} \oplus \parity{jk}
 \oplus \parity{ik})$ and $ (x_i \oplus x_k \oplus
e_{ik} \equiv \parity{ik})$. It holds that $\parity{} =
\parity{12} \oplus \parity{23} ... \oplus \parity{{n-1}n} \oplus \parity{1n} $.
 Thus, if $|X| > 3$, there must be an xor-cycle $\xorcycleS{X \backslash \set{x_j}}{(Y \cup
         \set{e_{ik}}) \backslash \set{e_{ij},e_{jk}}} {\parity{} \oplus \parity{ij} \oplus \parity{jk} \oplus \parity{ik}}$ in $\xorclauses \wedge \opttrans{\xorclauses}$.
 \qed
\end{proof}

\begin{lemma}
\label{Lem:OptEijUP}
Given a conjunction of xor-clauses $\xorclauses$ in 3-xor normal form and an xor-cycle
$\xorcycle{x_1,...,x_n}{y_1,...,y_n}{\parity{}}$, it holds that
$\xorclauses \wedge \opttrans{\xorclauses} \wedge (y_1 \equiv \parity{1}') \wedge ... \wedge (y_{n-1} \equiv \parity{n-1}') \UPDeriv
(y_n \equiv \parity{} \oplus \parity{1}' \oplus ... \oplus \parity{n-1}')$.
\end{lemma}
\begin{proof}
Assume a conjunction $\xorclauses$ in 3-xor normal form and an xor-cycle
$\xorcycleS{X}{Y}{\parity{}}$, $X{=}\Tuple{x_1,...,x_n}$, $Y{=}\Tuple{y_1,...,y_n}$ in $\xorclauses$. Base case $
n{=}3$: $\xorcycleS{X}{Y}{\parity{}} {=} (x_1 \XX x_2 \XX y_1 \equiv \parity{1})
\wedge (x_2 \XX x_3 \XX y_2 \equiv \parity{2}) \wedge (x_1 \XX x_3 \XX y_3 \equiv
\parity{3})$. It is clear that the algorithm in Fig.\ref{Fig:OptEij} adds the
xor-clause $ (y_1 \XX y_2 \XX y_3 \equiv \parity{})$ to $\opttrans{\xorclauses}$.
It follows that $ \xorclauses \wedge \opttrans{\xorclauses} \wedge (y_1 \equiv \parity{1}') \wedge
(y_2 \equiv \parity{2}') \UPDeriv (y_3 \equiv \parity{} \XX \parity{1}' \XX
\parity{2}')$. Induction hypothesis: Lemma~\ref{Lem:OptEijUP} holds
for all xor-cycles $\xorcycleS{X'}{Y'}{\parity{}''}$ such that $|X'|=n-1$.
Induction step: By Lemma~\ref{Lem:CycleCompress},  $\xorclauses \wedge
\opttrans{\xorclauses}$ has an xor-cycle $ \xorcycleS{X\backslash \set{x}}{(Y
        \cup \set{y}) \backslash \set{y',y''}}{\parity{}' \XX \parity{}}$
 and
an xor-clause $(y \XX y'
         \XX y'' \equiv \parity{}')$. There are two cases to consider:
\begin{itemize}
\item Case I: $y_n = y'$ or $y_n = y''$.
Without loss of generality, we can consider only the case where $y'' = y_n$ and
$y' = y_{n-1}$. In this case, the xor-clause $ (y \XX y_{n-1} \XX y_n \equiv
        \parity{}')$ is in $\xorclauses \wedge \opttrans{\xorclauses}$. By induction hypothesis $
\xorclauses \wedge \opttrans{\xorclauses} \wedge (y_1 \equiv \parity{1}') \wedge ... \wedge
(y_{n-2} \equiv \parity{n-2}') \UPDeriv (y \equiv \parity{} \oplus \parity{}'
        \oplus \parity{1}'
            \oplus ... \oplus \parity{n-2}') $. It follows that $ (y \XX
                y_{n-1} \XX y_n \equiv \parity{}') \wedge (y_{n-1} \equiv
                    \parity{n-1}') \wedge (y \equiv \parity{} \oplus \parity{}' \oplus
                        \parity{1}' \oplus ... \oplus \parity{n-2}')
 \UPDeriv (y_n \equiv \parity{} \XX \parity{1}' \XX ... \XX \parity{n-1}')$.

 \item Case II: $y_n \in Y \backslash \set{y',y''}$. 
Again without loss of generality, we can consider only the case where $y'' = y_{n-1} $, and
$y' = y_{n-2}$, so the xor-clause$ (y \XX y_{n-2} \XX y_{n-1} \equiv
        \parity{}')$ is in $\xorclauses \wedge \opttrans{\xorclauses}$.  
It follows that $ (y \XX y_{n-2} \XX y_{n-1} \equiv \parity{}') \wedge (y_{n-2} \equiv \parity{n-2}') \wedge (y_{n-1} \equiv \parity{n-1}') \UPDeriv (y \equiv \parity{}' \XX \parity{n-1}' \XX \parity{n-2}')$.
Then by induction hypothesis $\xorclauses \wedge \opttrans{\xorclauses} \wedge (y_1 \equiv \parity{1}')\wedge ... \wedge (y_{n-3} \equiv \parity{n-3}') \wedge (y \equiv
            \parity{}' \XX \parity{n-2}' \XX \parity{n-1}')
\UPDeriv  (y_n \equiv
        \parity{} \XX \parity{1}' \XX ... \XX \parity{n-1}')$.
\end{itemize}
\qed
\end{proof}

\begin{retheorem}{\ref{Thm:OptTransCorrectness}}
If $\xorclauses$ is an xor-clause conjunction in 3-xor normal form,
then $\opttrans{\xorclauses}$ is an \EC{}-simulation formula for $\xorclauses$.
\end{retheorem}
\begin{proof}
We first prove that the satisfying truth assignments of $\xorclauses$ are
exactly the ones of $\xorclauses \land \opttrans{\xorclauses}$ when projected
to $\VarsOf{\xorclauses}$. It holds by construction that
$\opttrans{\xorclauses} \subseteq \eijtrans{\xorclauses}$ and since
$\eijtrans{\xorclauses}$ is an \EC{}-simulation formula for $\xorclauses$ by
Theorem~\ref{Thm:EijTransCorrectness}, then the satisfying truth assignments of
$\xorclauses$ are exactly the ones of $\xorclauses \land
\opttrans{\xorclauses}$ when projected to $\VarsOf{\xorclauses}$.

We now prove that if $\IL$ is \EC-derivable from $\xorclauses \land \AL_1
   \land ... \land \AL_k$, then $\IL$ is \UP-derivable from $\xorclauses
            \land \eijtrans{\xorclauses} \land \AL_1 \land ...  \land
    \AL_k$.
 Assume that a literal $\IL$ is \EC-derivable from $ \psi = \xorclauses \land
 \AL_1 \land ... \land \AL_k$.  This implies that there is an
 \EC-derivation \ECderivation{} from $\psi$ and the literal $\IL$ is derived
 from the xor-clauses $C_1, ..., C_n$ in \ECderivation{} using one of the
 inference rules of \EC.  We prove by structural induction that $\IL$ is
 \UP-derivable from $\psi' = \xorclauses \wedge \eijtrans{\xorclauses} \land \AL_1 \land ...
 \land \AL_k$. The induction hypothesis is that there is a \UP-derivation
 \UPderivation{} from $\psi'$ such that the xor-clauses $C_1,...,C_n$ are in
 \UPderivation{}.  The inference rules \UnConflictRuleP{},  \TerBinRuleP{}, and
 \BinUnRuleP{} are special cases of \unitruleP{} and \unitruleN{}, 
  and \ConflictRuleP{} can be simulated by \TerUnRule{} and \UnConflictRuleP{},
 so it
 suffices to show that the inference rule \TerUnRule{} 
 can be simulated with \unitruleP{} and \unitruleN{}. 
  In the case that $\IL = (y \equiv \parity{} \oplus \parity{1}' \oplus ... \oplus \parity{n-1}') $ is derived using
the inference rule \TerUnRule{}, by Lemma~\ref{Lem:XorImply} there must be an
xor-cycle $\xorcycle{x_1,...,x_n}{y_1,...,y_{n-1},y}{\parity{}}$ in
$\xorclauses$
 such
 that for each $ y_i \in \set{y_1,...,y_{n-1}}$ it holds that $\psi \ECDeriv
 (y_i \equiv \parity{i}')$. By induction hypothesis it holds that $\psi' \UPDeriv (y_i \equiv \parity{i}')$ for each $ y_i
 \in \set{y_1,...,y_{n-1}}$. By Lemma~\ref{Lem:OptEijUP} and due to the existence of 
$\xorcycle{x_1,...,x_n}{y_1,...,y_{n-1},y}{\parity{}}$, it holds
that $ \xorclauses \wedge \opttrans{\xorclauses} \wedge (y_1 \equiv \parity{1}') \wedge ... \wedge (y_{n-1} \equiv \parity{n-1}') \UPDeriv
(y \equiv \parity{} \oplus \parity{1}' \oplus ... \oplus \parity{n-1}')$.
By induction it follows that $\IL$ is UP-derivable from $\xorclauses \wedge \opttrans{\xorclauses} \land \AL_1 \land ... \land \AL_k$.
\qed
\end{proof}

%% file: paper.bbl
\begin{thebibliography}{10}

\bibitem{LJN:CP2012}
Laitinen, T., Junttila, T., Niemel{\"a}, I.:
\newblock Classifying and propagating parity constraints.
\newblock In: Proc.\ CP 2012. Volume 7514 of LNCS., Springer (2012)  357--372

\bibitem{LJN:ICTAI2012full}
Laitinen, T., Junttila, T., Niemel\"a, I.:
\newblock Extending clause learning {SAT} solvers with complete parity
  reasoning (extended version).
\newblock arXiv document arXiv:1207.0988 [cs.LO] (2012)

\bibitem{Handbook:CDCL}
Marques-Silva, J., Lynce, I., Malik, S.:
\newblock Conflict-driven clause learning {SAT} solvers.
\newblock In: Handbook of Satisfiability.
\newblock IOS Press (2009)

\bibitem{Li:AAAI2000}
Li, C.M.:
\newblock Integrating equivalency reasoning into {Davis-Putnam} procedure.
\newblock In: Proc.\ AAAI/IAAI 2000, AAAI Press (2000)  291--296

\bibitem{Li:IPL2000}
Li, C.M.:
\newblock Equivalency reasoning to solve a class of hard {SAT} problems.
\newblock Information Processing Letters \textbf{76}(1--2) (2000)  75--81

\bibitem{BaumgartnerMassacci:CL2000}
Baumgartner, P., Massacci, F.:
\newblock The taming of the {(X)OR}.
\newblock In: Proc.~CL~2000. Volume 1861 of LNCS., Springer (2000)  508--522

\bibitem{Li:DAM2003}
Li, C.M.:
\newblock Equivalent literal propagation in the {DLL} procedure.
\newblock Discrete Applied Mathematics \textbf{130}(2) (2003)  251--276

\bibitem{HeuleMaaren:SAT2004}
Heule, M., van Maaren, H.:
\newblock Aligning {CNF}- and equivalence-reasoning.
\newblock In: Proc.\ SAT 2004. Volume 3542 of LNCS., Springer (2004)  145--156

\bibitem{HeuleEtAl:SAT2004}
Heule, M., Dufour, M., van Zwieten, J., van Maaren, H.:
\newblock March\_eq: Implementing additional reasoning into an efficient
  look-ahead {SAT} solver.
\newblock In: Proc.\ SAT 2004. Volume 3542 of LNCS., Springer (2004)  345--359

\bibitem{Chen:SAT2009}
Chen, J.:
\newblock Building a hybrid {SAT} solver via conflict-driven, look-ahead and
  {XOR} reasoning techniques.
\newblock In: Proc.\ SAT 2009. Volume 5584 of LNCS., Springer (2009)  298--311

\bibitem{SoosEtAl:SAT2009}
Soos, M., Nohl, K., Castelluccia, C.:
\newblock Extending {SAT} solvers to cryptographic problems.
\newblock In: Proc.\ SAT 2009. Volume 5584 of LNCS., Springer (2009)  244--257

\bibitem{LJN:ECAI2010}
Laitinen, T., Junttila, T., Niemel{\"a}, I.:
\newblock Extending clause learning {DPLL} with parity reasoning.
\newblock In: Proc.\ ECAI 2010, IOS Press (2010)  21--26

\bibitem{Soos}
Soos, M.:
\newblock Enhanced {Gaussian} elimination in {DPLL}-based {SAT} solvers.
\newblock In: Pragmatics of SAT, Edinburgh, Scotland, GB (July 2010)  1--1

\bibitem{LJN:ICTAI2011}
Laitinen, T., Junttila, T., Niemel{\"a}, I.:
\newblock Equivalence class based parity reasoning with {DPLL(XOR)}.
\newblock In: ICTAI, IEEE (2011)  649--658

\bibitem{LJN:SAT2012}
Laitinen, T., Junttila, T., Niemel{\"a}, I.:
\newblock Conflict-driven {XOR}-clause learning.
\newblock In: Proc. SAT 2012. Volume 7317 of Lecture Notes in Computer
  Science., Springer (2012)  383--396

\bibitem{NieuwenhuisEtAl:JACM06}
Nieuwenhuis, R., Oliveras, A., Tinelli, C.:
\newblock Solving {SAT} and {SAT} modulo theories: From an abstract
  {D}avis-{P}utnam-{L}ogemann-{L}oveland procedure to {DPLL(T)}.
\newblock Journal of the ACM \textbf{53}(6) (2006)  937--977

\bibitem{Handbook:SMT}
Barrett, C., Sebastiani, R., Seshia, S.A., Tinelli, C.:
\newblock Satisfiability modulo theories.
\newblock In: Handbook of Satisfiability.
\newblock IOS Press (2009)

\end{thebibliography}
